\documentclass[pra,twocolumn,showpacs,showkeys,byrevtex]{revtex4}%
\pdfoutput=1
\usepackage{amsfonts}
\usepackage{amsmath}
\usepackage{amssymb}
\usepackage{amsthm}
\usepackage{graphicx}
\usepackage{verbatim}%
\providecommand{\U}[1]{\protect\rule{.1in}{.1in}}
\newtheorem{theorem}{Theorem}

\newtheorem{definition}{Definition}
\newtheorem{example}{Example}

\newtheorem{lemma}{Lemma}

\begin{document}

\title{Entanglement-assisted quantum convolutional coding}
\author{Mark M. Wilde}
\email{mwilde@gmail.com}
\author{Todd A. Brun}
\email{tbrun@usc.edu}
\affiliation{Communication Sciences Institute, Department of
Electrical Engineering, University of Southern California, Los Angeles,
California 90089 USA}

\begin{abstract}
We show how to protect a stream of quantum information from decoherence
induced by a noisy quantum communication channel. We exploit preshared
entanglement and a convolutional coding structure to develop a theory of
entanglement-assisted quantum convolutional coding. Our construction produces
a Calderbank-Shor-Steane (CSS) entanglement-assisted quantum convolutional
code from two arbitrary classical binary convolutional codes. The rate and
error-correcting properties of the classical convolutional codes directly
determine the corresponding properties of the resulting entanglement-assisted
quantum convolutional code. We explain how to encode our
CSS\ entanglement-assisted quantum convolutional codes starting from a stream
of information qubits, ancilla qubits, and shared entangled bits.
\end{abstract}


\maketitle

\section{Introduction}

Quantum error correction theory
\cite{PhysRevA.52.R2493,PhysRevLett.77.793,PhysRevA.54.1098,thesis97gottesman,PhysRevLett.78.405,ieee1998calderbank}%
\ stands as the pivotal theoretical tool that will make reliable quantum
computing and quantum communication possible. Any future quantum information
processing device will operate faithfully only if it employs an error
correction scheme. This scheme can be an active scheme
\cite{thesis97gottesman},\ a passive scheme
\cite{PhysRevLett.79.3306,mpl1997zanardi,PhysRevLett.81.2594},\ or a
combination of both techniques
\cite{kribs:180501,qic2006kribs,poulin:230504,isit2007brun,arxiv2007brun}.

Mermin proclaims it a \textquotedblleft miracle\textquotedblright\ that
quantum error correction is even possible \cite{book2007mermin}.\ Various
obstacles such as the no-cloning theorem \cite{nat1982}, measurement
destroying a quantum state, and continuous quantum errors seem to pose an
insurmountable barrier to a protocol for quantum error correction. Despite
these obstacles, Shor demonstrated the first quantum error-correcting code
that reduces the negative effects of decoherence on a quantum bit
\cite{PhysRevA.52.R2493}. Shor's code overcame all of the above difficulties
and established the basic principles for constructing a general theory of
quantum error correction
\cite{thesis97gottesman,PhysRevLett.78.405,ieee1998calderbank}.

Gottesman formalized the theory of quantum block coding by establishing the
stabilizer formalism \cite{thesis97gottesman}. The stabilizer formalism allows
one to import self-orthogonal classical block codes for use in quantum error
correction \cite{ieee1998calderbank}. This technique has the benefit of
exploiting the large body of research on classical coding theory
\cite{book1983code}\ for use in quantum error correction, but the
self-orthogonality constraint limits the classical block codes that we can import.

Bowen was the first to extend the stabilizer formalism by providing an example
of a code that exploits entanglement shared between a sender and a receiver
\cite{PhysRevA.66.052313}. The underlying assumption of Bowen's code is that
the sender and receiver share a set of noiseless ebits (entangled
qubits)\ before quantum communication begins. Many quantum protocols such as
teleportation \cite{PhysRevLett.70.1895}\ and superdense coding
\cite{PhysRevLett.69.2881}\ are \textquotedblleft
entanglement-assisted\textquotedblright\ protocols because they assume that
noiseless ebits are available.

Brun, Devetak, and Hsieh generalized Bowen's example by constructing a theory
of stabilizer codes that employs ancilla qubits and shared ebits for encoding
a quantum error-correcting code \cite{arx2006brun,science2006brun}. The
so-called entanglement-assisted stabilizer formalism subsumes the stabilizer
formalism as the theory of active quantum error correction.

The major benefit of the entanglement-assisted stabilizer formalism\ is that
we can construct an entanglement-assisted quantum code from two arbitrary
classical binary block codes or from an arbitrary classical quaternary block
code. The rates and error-correcting properties of the classical codes
translate to the resulting quantum codes. The entanglement-assisted stabilizer
formalism may be able to reduce the problem of finding high-performance
quantum codes approaching the quantum capacity
\cite{PhysRevA.55.1613,capacity2002shor,ieee2005dev,PhysRevLett.83.3081,ieee2002bennett}%
\ to the problem of finding good classical linear codes approaching the
classical capacity \cite{book1991cover}.

Another extension of the theory of quantum error correction protects a
potentially-infinite stream of quantum information against the corruption
induced by a noisy quantum communication channel
\cite{PhysRevLett.91.177902,arxiv2004olliv,isit2006grassl,ieee2006grassl,ieee2007grassl,isit2005forney,ieee2007forney}%
. These quantum convolutional codes possess several advantages over quantum
block codes. A quantum convolutional code typically has lower encoding and
decoding complexity and superior code rate when compared to a block code that
protects the same number of information qubits \cite{ieee2007forney}.

Forney et al. have determined a method for importing an arbitrary classical
self-orthogonal quaternary code for use as a quantum convolutional code
\cite{isit2005forney,ieee2007forney}.\ The technique is similar to that for
importing a classical block code as a quantum block code
\cite{ieee1998calderbank}. One limitation of this technique is that the
self-orthogonality constraint is more restrictive in the convolutional
setting. Each generator for the quantum convolutional code must commute not
only with the other generators, but it must commute also with any arbitrary
shift of itself and any arbitrary shift of the other generators. Forney et al.
performed specialized searches to determine classical quaternary codes that
satisfy the restrictive self-orthogonality constraint \cite{ieee2007forney}.

In this paper, we develop a theory of entanglement-assisted quantum
convolutional coding for a broad class of codes. Our major result is that we
can produce an entanglement-assisted quantum convolutional code from two
\textit{arbitrary} classical binary convolutional codes. The resulting quantum
convolutional codes admit a Calderbank-Shor-Steane (CSS)\ structure
\cite{PhysRevA.54.1098,PhysRevLett.77.793,book2000mikeandike}. The rates and
error-correcting properties of the two binary classical convolutional codes
directly determine the corresponding properties of the entanglement-assisted
quantum convolutional code.

Our techniques for encoding and decoding are also an expansion of previous
techniques from quantum convolutional coding theory. Previous techniques for
encoding and decoding include finite-depth operations only. A finite-depth
operation propagates errors to a finite number of neighboring qubits in the
qubit stream. We introduce an infinite-depth operation to the set of
shift-invariant Clifford operations and explain it in detail in
Section~\ref{sec:infinite-depth-ops}. We must be delicate when using
infinite-depth operations because they can propagate errors to an infinite
number of neighboring qubits in the qubit stream. We explain our assumptions
in detail in Section~\ref{sec:eaqcc-iefd}\ for including infinite-depth
operations in our entanglement-assisted quantum convolutional codes. An
infinite-depth operation gives more flexibility when designing encoding
circuits---similar to the way in which an infinite-impulse response filter
gives more flexibility in the design of classical convolutional circuits. It
also is the key operation enabling us to import arbitrary classical
convolutional codes for entanglement-assisted quantum coding.

Our CSS\ entanglement-assisted quantum convolutional codes divide into two
classes based on certain properties of the classical codes from which we
produce them. These properties of the classical codes determine the structure
of the encoding and decoding circuit for the code, and the structure of the
encoding and decoding circuit in turn determines the class of the
entanglement-assisted quantum convolutional code.

\begin{enumerate}
\item Codes in the first class admit both a finite-depth encoding and decoding circuit.

\item Codes in the second class have an encoding circuit that employs both
finite-depth and infinite-depth operations. Their decoding circuits have
finite-depth operations only.
\end{enumerate}

We structure our work as follows. Section~\ref{sec:rev-conv-stab} reviews the
stabilizer formalism for quantum block codes, entanglement-assisted quantum
codes, and convolutional stabilizer codes. We review the important isomorphism
that allows us to work with matrices of binary polynomials rather than
infinite tensor products of Pauli matrices. Section~\ref{sec:finite-depth-ops}%
\ reviews finite-depth Clifford operations for use in encoding and decoding
\cite{isit2006grassl,ieee2006grassl,ieee2007grassl}. We outline the operation
of an entanglement-assisted quantum convolutional code and present our main
theorem in Section~\ref{sec:eaqcc}. This theorem shows how to produce a CSS
entanglement-assisted quantum convolutional code from two arbitrary classical
binary convolutional codes. The theorem gives the rate and error-correcting
properties of a CSS\ entanglement-assisted quantum convolutional code as a
function of the parameters of the classical convolutional codes.
Section~\ref{sec:eaqcc-fefd} completes the proof of the theorem for our first
class of entanglement-assisted quantum convolutional codes. In
Section~\ref{sec:infinite-depth-ops}, we introduce an infinite-depth encoding
operation to the set of shift-invariant Clifford operations and discuss its
effect on both the stabilizer and the logical operators for the information
qubits. Section~\ref{sec:eaqcc-iefd} completes the proof of our theorem for
the second class of entanglement-assisted quantum convolutional codes. We
discuss the implications of the assumptions for the different classes of
entanglement-assisted quantum convolutional codes while developing the
constructions. Our hope is that our theory will produce high-performance
quantum convolutional codes by importing high-performance classical
convolutional codes.

\section{Review of the Stabilizer Formalism}

\label{sec:rev-conv-stab}The stabilizer formalism is a mathematical framework
for quantum error correction \cite{PhysRevA.54.1862,thesis97gottesman}. This
framework has many similarities with classical coding theory, and it is even
possible to import a classical code for use in quantum error correction by
employing the CSS\ construction
\cite{PhysRevA.54.1098,PhysRevLett.77.793,book2000mikeandike}. We briefly
review the stabilizer theory for quantum block codes, entanglement-assisted
quantum block codes, and quantum convolutional codes (see Refs.
\cite{ieee2007forney,arx2007wilde}\ for a more detailed review).

\subsection{Stabilizer Formalism for Quantum Block Codes}

\label{sec:stabilizer-review}The following four matrices%
\[
I\equiv%
\begin{bmatrix}
1 & 0\\
0 & 1
\end{bmatrix}
,\ X\equiv%
\begin{bmatrix}
0 & 1\\
1 & 0
\end{bmatrix}
,\ Y\equiv%
\begin{bmatrix}
0 & -i\\
i & 0
\end{bmatrix}
,\ Z\equiv%
\begin{bmatrix}
1 & 0\\
0 & -1
\end{bmatrix}
,
\]
in the Pauli group $\Pi=\left\{  I,X,Y,Z\right\}  $\ are the most important in
formulating a quantum error-correcting code. Two crucial properties of these
matrices are useful: each matrix in $\Pi$ has eigenvalues equal to $+1$ or
$-1$, and any two matrices in $\Pi$ either commute or anticommute. Matrices in
$\Pi$ act on a two-dimensional complex vector, or equivalently, a single qubit.

In general, a quantum error-correcting code uses $n$ physical qubits to
protect a smaller set of information qubits against decoherence or quantum
noise. An $n$-qubit quantum error-correcting code employs elements of the
Pauli group $\Pi^{n}$. The Pauli group $\Pi^{n}$ consists of $n$-fold tensor
products of Pauli matrices:%
\begin{equation}
\Pi^{n}=\left\{
\begin{array}
[c]{c}%
e^{i\phi}A_{1}\otimes\cdots\otimes A_{n}:\forall j\in\left\{  1,\ldots
,n\right\}  ,\\
A_{j}\in\Pi,\ \ \phi\in\left\{  0,\pi/2,\pi,3\pi/2\right\}
\end{array}
\right\}  .
\end{equation}
We liberally omit the tensor product symbol in what follows so that
$A_{1}\cdots A_{n}\equiv A_{1}\otimes\cdots\otimes A_{n}$. The above two
crucial properties for the single-qubit Pauli group $\Pi$ still hold for the
Pauli group $\Pi^{n}$ (up to an irrelevant phase for the eigenvalue property).
Matrices in $\Pi^{n}$ act on a $2^{n}$-dimensional complex vector, or
equivalently, an $n$-qubit quantum register.

We can phrase the theory of quantum error correction in purely mathematical
terms using elements of $\Pi^{n}$. Consider a matrix $g_{1}\in\Pi^{n}$ that is
not equal to $\pm I$. Matrix $g_{1}$ then has two eigenspaces each of size
$2^{n-1}$. We can identify one eigenspace with the eigenvalue $+1$ and the
other eigenspace with eigenvalue $-1$. Consider a matrix $g_{2}\in\Pi^{n}$
different from $g_{1}$ that commutes with $g_{1}$. Matrix $g_{2}$ also has two
eigenspaces each of size $2^{n-1}$ and identified similarly by its eigenvalues
$\pm1$. Both $g_{1}$ and $g_{2}$ have simultaneous eigenspaces because they
commute. These matrices together have four different eigenspaces, each of size
$2^{n-2}$ and identified by the eigenvalues $\pm1,\pm1$ of $g_{1}$ and $g_{2}$
respectively. We can continue this process of adding more commuting and
independent matrices to a set $\mathcal{S}$. The matrices in $\mathcal{S}$ are
independent in the sense that no matrix in $\mathcal{S}$ is a product of two
or more other matrices in $\mathcal{S}$. Adding more matrices from $\Pi^{n}%
$\ to $\mathcal{S}$ continues to divide the eigenspaces of matrices in
$\mathcal{S}$. In general, suppose $\mathcal{S}$ consists of $n-k$ independent
and commuting matrices $g_{1}$, \ldots, $g_{n-k}\in\Pi^{n}$. These $n-k$
matrices then have $2^{n-k}$ different eigenspaces each of size $2^{k}$ and
identified by the eigenvalues $\pm1$, \ldots, $\pm1$\ of $g_{1}$, \ldots,
$g_{n-k}$ respectively. Consider that the Hilbert space of $k$ qubits has size
$2^{k}$. A dimension count immediately suggests that we can encode $k$ qubits
into one of the eigenspaces of $\mathcal{S}$. We typically encode these $k$
qubits into the simultaneous $+1$-eigenspace of $g_{1}$, \ldots, $g_{n-k}$.
This eigenspace is the \textit{codespace}. An $\left[  n,k\right]  $ quantum
error-correcting code encodes $k$ information qubits into the simultaneous
$+1$-eigenspace of $n-k$ matrices $g_{1}$, \ldots, $g_{n-k}\in\Pi^{n}$. The
rate of an $\left[  n,k\right]  $ code is the ratio of information qubits to
physical qubits: $k/n$.%
\begin{figure}
[ptb]
\begin{center}
\includegraphics[
natheight=3.386600in,
natwidth=8.973300in,
height=1.292in,
width=3.4039in
]%
{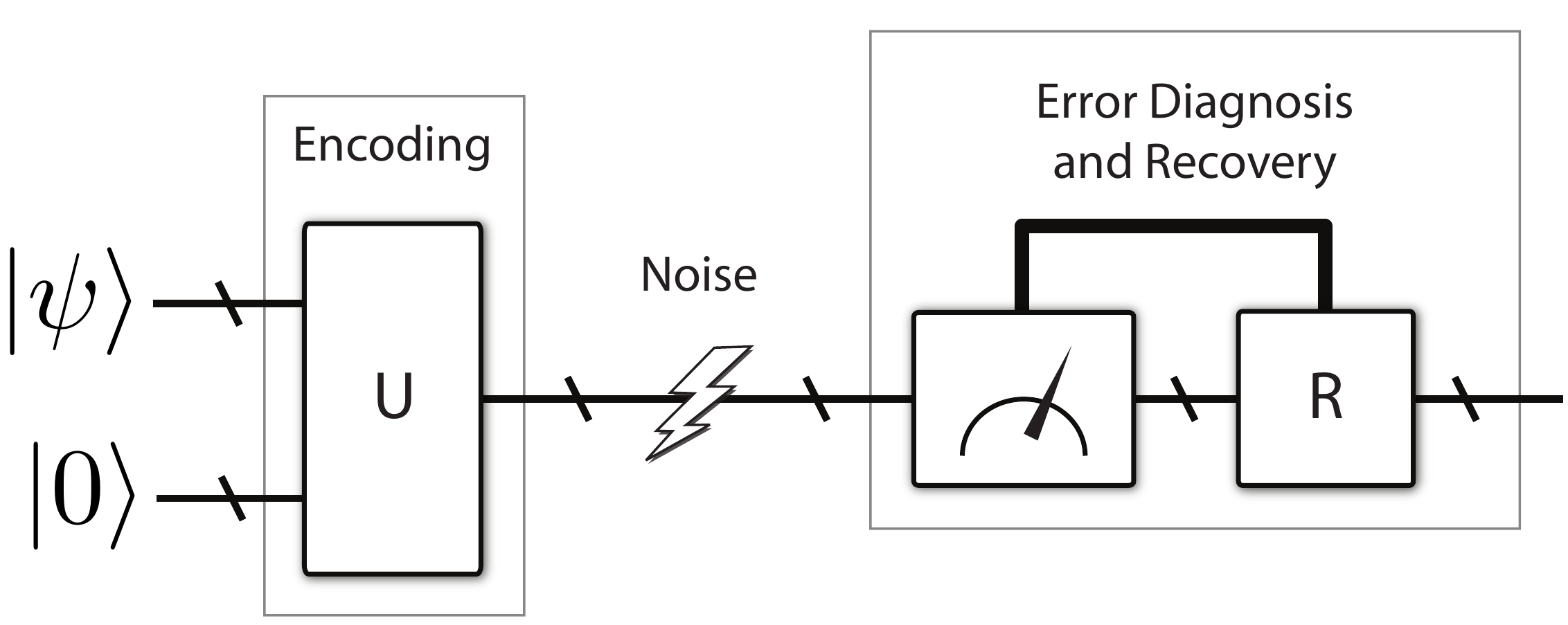}%
\caption{The operation of a general stabilizer code. Thin lines denote quantum
information and thick lines denote classical information. Slanted bars denote
multiple qubits. A sender employs a unitary encoding operation $U$ to encode a
set of information qubits in the state $\left\vert \psi\right\rangle $ with
the help of ancilla qubits each in the state $\left\vert 0\right\rangle $. The
sender transmits the encoded qubits over the noisy quantum communication
channel. The receiver performs quantum measurements to diagnose which error
has occurred. He finally performs a recovery operation $R$ to reverse the
error from the channel.}%
\label{fig:stabilizer}%
\end{center}
\end{figure}

The operation of an $\left[  n,k\right]  $ quantum error-correcting code
consists of four steps. Figure~\ref{fig:stabilizer}\ highlights these steps.
First, a unitary operation $U$ encodes $k$ qubits and $n-k$ ancilla qubits
into the simultaneous $+1$-eigenspace of the matrices $g_{1}$, \ldots,
$g_{n-k}$. The sender transmits the $n$ encoded qubits by using the noisy
quantum communication channel $n$ times. The receiver performs quantum
measurements of the $n-k$ matrices $g_{1}$, \ldots, $g_{n-k}$. These
measurements learn only about errors that may occur and do not disturb the
encoded quantum information. Each measurement gives a bit result equal to $+1$
or $-1$, and the result of all the measurements is to project the $n$-qubit
quantum register into one of the $2^{n-k}$ different eigenspaces of $g_{1}$,
\ldots, $g_{n-k}$. Suppose that no error occurs. Then the measurements project
the $n$ qubits into the simultaneous $+1$-eigenspace and return a bit vector
consisting of $n-k$ ones. Now suppose that a quantum error in an error set
$\mathcal{E}$ occurs. The error takes the encoded quantum state out of the
codespace and into one of the other $2^{n-k}-1$ orthogonal eigenspaces. The
measurements can detect that an error has occurred because the result of the
measurements is a bit vector differing from the all ones vector. The receiver
may be able to identify uniquely which error has occurred if it satisfies the
following quantum error correction conditions:%
\[
\forall E_{a},E_{b}\in\mathcal{E\ \ }\exists\ g_{i}\in\mathcal{S}:\left\{
g_{i},E_{a}^{\dag}E_{b}\right\}  =0\text{ \ or \ }E_{a}^{\dag}E_{b}%
\in\mathcal{S}.
\]
The first condition states that errors are detectable if they anticommute with
one of the generators in $\mathcal{S}$, and the second condition states that
errors have no effect on the encoded state if they are in $\mathcal{S}$. If
the receiver can identify which error occurs, he can then apply unitary
operation $R$\ that is the inverse of the error. He finally performs a
decoding unitary that decodes the $k$ information qubits.

We comment briefly on the encoding operation $U$. The encoding operation $U$
is a special type of unitary matrix called a Clifford operation. A Clifford
operation $U$\ is one that preserves elements of the Pauli group under
conjugation:\ $A\in\Pi^{n}\Rightarrow UAU^{\dag}\in\Pi^{n}$. The CNOT\ gate,
the Hadamard gate $H$, and the phase gate $P$ suffice to implement any unitary
matrix in the Clifford group \cite{thesis97gottesman}. A quantum code with the
CSS\ structure needs only the CNOT\ and Hadamard gates for encoding and
decoding. The matrix for the CNOT\ gate acting on two qubits is%
\begin{equation}
\text{CNOT}=%
\begin{bmatrix}
1 & 0 & 0 & 0\\
0 & 1 & 0 & 0\\
0 & 0 & 0 & 1\\
0 & 0 & 1 & 0
\end{bmatrix}
,
\end{equation}
the matrix for the Hadamard gate $H$ acting on a single qubit is%
\begin{equation}
H=\frac{1}{\sqrt{2}}%
\begin{bmatrix}
1 & 1\\
1 & -1
\end{bmatrix}
,
\end{equation}
and the matrix for the phase gate $P$ acting on a single qubit is%
\begin{equation}
P=%
\begin{bmatrix}
1 & 0\\
0 & i
\end{bmatrix}
.
\end{equation}
For the CNOT\ gate, the first qubit is the \textquotedblleft
control\textquotedblright\ qubit and the second qubit is the \textquotedblleft
target\textquotedblright\ qubit. The standard basis for elements of the
two-qubit Pauli group $\Pi^{2}$ is as follows%
\begin{equation}%
\begin{array}
[c]{cc}%
Z & I\\
I & Z\\
X & I\\
I & X
\end{array}
,
\end{equation}
because any element of $\Pi^{2}$ is a product of the above four matrices up to
an irrelevant phase. The standard basis for $\Pi^{1}$ is $X$ and $Z$ for the
same reasons. The CNOT\ gate transforms the standard basis of $\Pi^{2}$ under
conjugation as follows%
\begin{equation}%
\begin{array}
[c]{cc}%
Z & I\\
I & Z\\
X & I\\
I & X
\end{array}
\rightarrow%
\begin{array}
[c]{cc}%
Z & I\\
Z & Z\\
X & X\\
I & X
\end{array}
,
\end{equation}
where the first qubit is the control and the second qubit is the target. The
Hadamard gate $H$ transforms the standard basis of $\Pi^{1}$ under conjugation
as follows:%
\begin{equation}%
\begin{array}
[c]{c}%
Z\\
X
\end{array}
\rightarrow%
\begin{array}
[c]{c}%
X\\
Z
\end{array}
,
\end{equation}
and the phase gate $P$ transforms the standard basis as follows:%
\begin{equation}%
\begin{array}
[c]{c}%
Z\\
X
\end{array}
\rightarrow%
\begin{array}
[c]{c}%
Z\\
Y
\end{array}
.
\end{equation}
Appendix of Ref.~\cite{arx2007wilde} details an algorithm that determines an
encoding circuit consisting of CNOT, $H$, and $P$ gates for any stabilizer
code or any entanglement-assisted stabilizer code (we review
entanglement-assisted codes in the next section).

Another aspect of the theory of quantum error correction is later useful for
our purposes in quantum convolutional coding. This aspect concerns the
information qubits and the operators that change them. Consider that the
initial unencoded state of a quantum error-correcting code is a simultaneous
+1-eigenstate of the matrices $Z_{k+1},\ldots,Z_{n}$ where $Z_{i}$ has a $Z$
matrix operating on qubit $i$ and the identity $I$ on all other qubits.
Therefore, the matrices $Z_{k+1},\ldots,Z_{n}$ constitute a stabilizer for the
unencoded state. The initial unencoded logical operators for the information
qubits are $Z_{1},X_{1},\ldots,Z_{k},X_{k}$. The encoding operation $U$
rotates the unencoded stabilizer matrices $Z_{k+1},\ldots,Z_{n}$ and the
unencoded logical operators $Z_{1},X_{1},\ldots,Z_{k},X_{k}$ to the encoded
stabilizer $\bar{Z}_{k+1}$, \ldots, $\bar{Z}_{n}$ and the encoded logical
operators $\bar{Z}_{1},\bar{X}_{1},\ldots,\bar{Z}_{k},\bar{X}_{k}$
respectively. The encoded matrices $\bar{Z}_{k+1},\ldots,\bar{Z}_{n}$ are
respectively equivalent to the matrices $g_{1}$, \ldots, $g_{n-k}$ in the
above discussion. The encoded operators obey the same commutation relations as
their unencoded counterparts. We would violate the uncertainty principle if
this invariance does not hold. Therefore, each of the encoded logical
operators commutes with elements of the stabilizer $\mathcal{S}$. Let
$A$\ denote an arbitrary logical operator from the above set and let $\bar
{Z}_{i}$ denote an arbitrary element of the stabilizer $\mathcal{S}$. The
operator $A\bar{Z}_{i}$ (or equivalently $\bar{Z}_{i}A$) is an equivalent
logical operator because $A\bar{Z}_{i}$ and $A$ have the same effect on an
encoded state $\left\vert \bar{\psi}\right\rangle $:%
\begin{equation}
\bar{Z}_{i}A\left\vert \bar{\psi}\right\rangle =A\bar{Z}_{i}\left\vert
\bar{\psi}\right\rangle =A\left\vert \bar{\psi}\right\rangle .
\end{equation}
We make extensive use of the above fact in our work.

The logical operators also provide a useful way to characterize the
information qubits. Gottesman showed that the logical operators for the
information qubits provide a straightforward way to characterize the
information qubits as they progress through a quantum circuit
\cite{thesis97gottesman}. As an example of this technique, he develops quantum
teleportation in the stabilizer formalism. The logical operators at the
beginning of the protocol are $X_{1}$ and $Z_{1}$ and become $X_{3}$ and
$Z_{3}$ at the end of the protocol. The quantum information in qubit one
teleports to qubit three because the logical operators act on only qubit three
at the end of the protocol. We use the same idea throughout this paper to
determine if our decoding circuits have truly decoded the information qubits.

It is possible to produce a stabilizer code from two classical binary block
codes by employing the CSS\ construction. The elements of the stabilizer group
of a CSS\ stabilizer code commute if and only if the codewords of one
classical code are orthogonal to the codewords of the other classical code
with respect to the binary inner product. The codes that we can import must
satisfy this condition because the commuting condition is essential in
formulating a quantum code. The entanglement-assisted stabilizer formalism
finds a clever way around this restriction by exploiting entanglement shared
between sender and receiver.

\subsection{Entanglement-Assisted Stabilizer Formalism for Quantum Block
Codes}

The entanglement-assisted stabilizer formalism is a significant extension of
the standard stabilizer formalism that incorporates shared entanglement as a
resource for encoding \cite{arx2006brun,science2006brun}. Several references
provide a review of this technique and generalizations of the basic theory to
block \cite{luo:010303} and convolutional \cite{arx2007wilde}\ entanglement
distillation protocols, continuous-variable codes \cite{pra2007wildeEA}, and
entanglement-assisted operator codes for discrete-variable
\cite{isit2007brun,arxiv2007brun}\ and continuous-variable systems
\cite{prep2007wildeEAOQEC}.

An entanglement-assisted code employs ebits or Bell states in addition to
ancilla qubits for quantum redundancy. We express the state $\left\vert
\Phi^{+}\right\rangle $\ of an ebit shared between a sender Alice and a
receiver Bob as follows:%
\begin{equation}
\left\vert \Phi^{+}\right\rangle \equiv\frac{\left\vert 00\right\rangle
^{AB}+\left\vert 11\right\rangle ^{AB}}{\sqrt{2}}.
\end{equation}

The advantage of the entanglement-assisted stabilizer formalism is that it
allows us to exploit the error-correcting properties of an arbitrary set of
Pauli matrices. They do not necessarily have to form a commuting set. In
particular, this construction allows us to produce a quantum block code from
two arbitrary classical binary block codes by employing the CSS\ construction.
Two high-performance classical block codes lead to a high-performance
entanglement-assisted quantum code. The entanglement-assisted method allows us
to exploit the full error-correcting power of classical coding theory.

An $\left[  n,k;c\right]  $\ entanglement-assisted code uses $c$ ebits and
$n-k-c$ ancilla qubits to encode $k$ information qubits. It operates as
follows. The sender and receiver share $c$ ebits before quantum communication
begins. The sender encodes her $k$ information qubits with the help of $n-k-c$
ancilla qubits and her half of the $c$ ebits. She performs an encoding
operation $U$\ on her $n$ qubits and sends them over a noisy quantum
communication channel. The noisy channel affects these $n$ qubits only and
does not affect the receiver's half of the $c$ ebits. The receiver combines
his half of the $c$ ebits with those he receives from the noisy quantum
channel. He performs measurements on all $n+c$ qubits to diagnose an error
that may occur on the $n$ qubits. He learns which error occurs and performs a
recovery operation that eliminates the error. Figure~\ref{fig:EA-code}%
\ illustrates the operation of an entanglement-assisted stabilizer code.%
\begin{figure}
[ptb]
\begin{center}
\includegraphics[
natheight=4.266100in,
natwidth=10.639800in,
height=1.3785in,
width=3.4238in
]%
{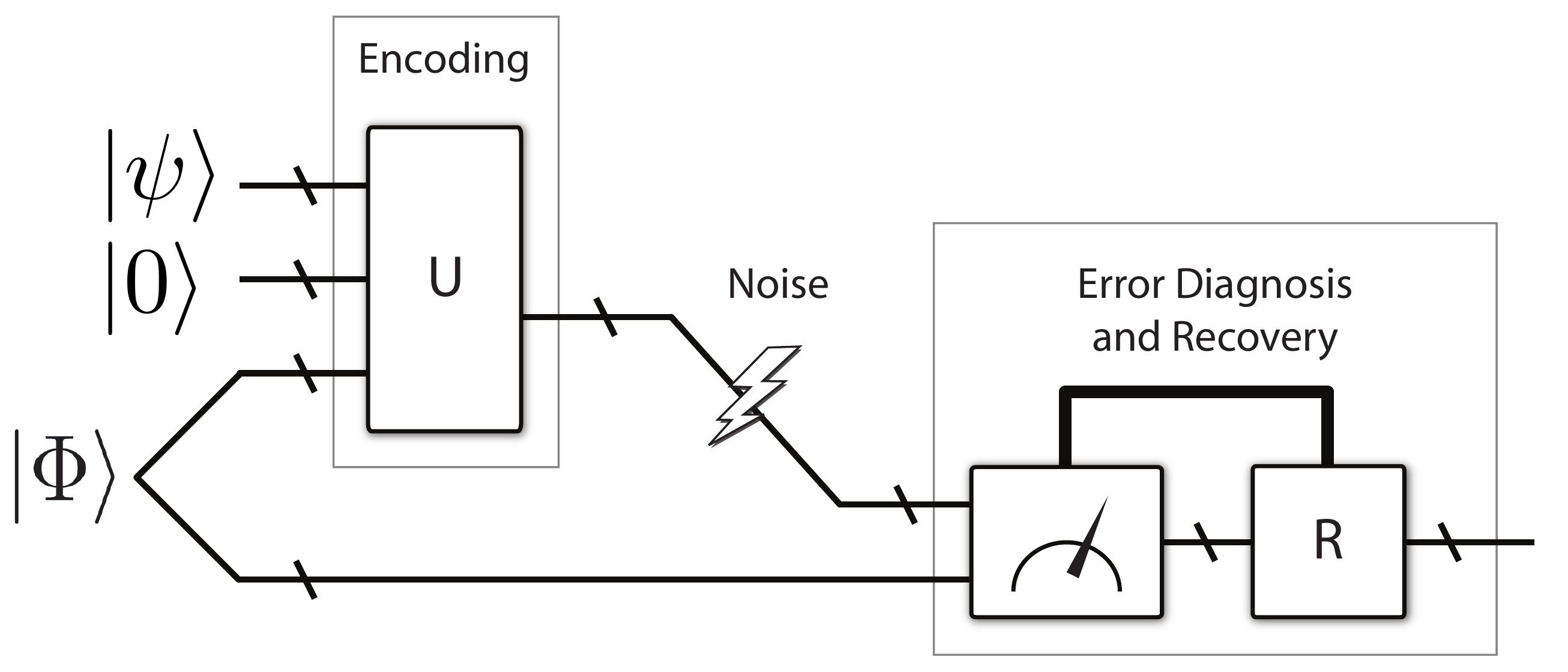}%
\caption{The operation of a general entanglement-assisted stabilizer code. The
sender encodes a set of information qubits with the help of ancilla qubits and
her half of a set of shared ebits. She sends her encoded qubits over a noisy
quantum communication channel. The entanglement-assisted communication
paradigm assumes that the receiver's half of the shared ebits remain noiseless
throughout this process. The receiver combines the noisy qubits with his half
of the shared ebits. He performs measurements on all of the qubits to diagnose
which error occurs and reverses the effect of this error by performing a
recovery operation.}%
\label{fig:EA-code}%
\end{center}
\end{figure}

Suppose we have an arbitrary set of Pauli matrices in $\Pi^{n}$ whose
error-correcting properties we would like to exploit. We do not necessarily
know beforehand how many ebits we require for the Pauli matrices to form a
commuting set, and we would like a method to determine the minimum number of
ebits. Several methods exist
\cite{arx2006brun,science2006brun,isit2007brun,arxiv2007brun,arx2007wilde},
but the algorithm in the Appendix of Ref.~\cite{arx2007wilde} determines the
minimum number of ebits required for the code, the encoding and decoding
circuit for the code, and the measurements the receiver performs to diagnose
errors. It essentially \textquotedblleft kills three birds with one
stone.\textquotedblright\ The algorithms we employ in this work are similar to
the algorithm in Ref.~\cite{arx2007wilde}, but they are quite a bit more
complicated because of the convolutional nature of our codes.

\subsubsection{Rate of an Entanglement-Assisted Quantum Code}

We can interpret the rate of an entanglement-assisted quantum convolutional
code in three different ways \cite{science2006brun,arx2006brun,arx2007wilde}.
Suppose that an entanglement-assisted quantum code encodes $k$ qubits in $n$
qubits with the help of $c$ ebits.

\begin{enumerate}
\item The \textquotedblleft entanglement-assisted\textquotedblright\ rate
assumes that entanglement shared between sender and receiver is free. Bennett
et al. make this assumption when deriving the entanglement-assisted capacity
of a quantum channel for sending quantum information
\cite{PhysRevLett.83.3081,ieee2002bennett}. The entanglement-assisted rate for
the above example is $k/n$.

\item The \textquotedblleft trade-off' rate assumes that entanglement is not
free and a rate pair determines performance. The first number in the pair is
the number of noiseless qubits generated per channel use, and the second
number in the pair is the number of ebits consumed per channel use. The rate
pair for the above example is $\left(  k/n,c/n\right)  $. Quantum information
theorists have computed asymptotic trade-off curves that bound the rate region
in which achievable rate pairs lie \cite{arx2005dev}. Brun et al.'s
construction for an entanglement-assisted quantum block code minimizes the
number $c$ of ebits given a fixed number $k$ and $n$ of information qubits and
encoded qubits respectively \cite{arx2006brun,science2006brun}.

\item The \textquotedblleft catalytic rate\textquotedblright\ assumes that
bits of entanglement are built up at the expense of transmitted qubits
\cite{arx2006brun,science2006brun}. A noiseless quantum channel or the encoded
use of noisy quantum channel are two different ways to build up entanglement
between a sender and receiver. The catalytic rate for the above code is
$\left(  k-c\right)  /n$.
\end{enumerate}

Which interpretation is most reasonable depends on the context in which we use
the code. In any case, the parameters $n$, $k$, and $c$ ultimately govern
performance, regardless of which definition of the rate we use to interpret
that performance.

\subsection{Stabilizer Formalism for Quantum Convolutional Codes}

We review the theory of convolutional stabilizer codes by considering a set of
Pauli matrices that stabilize a stream of encoded qubits. We follow with the
most important part of this review---the isomorphism from the set of Pauli
sequences to the module over the ring of binary polynomials
\cite{arxiv2004olliv,isit2006grassl,ieee2007forney}. We name it the
Pauli-to-binary (P2B)\ isomorphism. The P2B\ isomorphism is important because
it is easier to perform manipulations with vectors of binary polynomials than
with Pauli sequences.

We review the notation and basic definitions first. A Pauli sequence
$\mathbf{A}$\ is a countably infinite tensor product of Pauli matrices $A_{i}%
$:%
\[
\mathbf{A}=%
{\displaystyle\bigotimes\limits_{i=0}^{\infty}}
\ A_{i}.
\]
The weight of a Pauli sequence is the number of Pauli matrices in the
countably-infinite tensor product that are not equal to the identity matrix. A
Pauli sequence has finite support if its weight is finite. Let $\Pi
^{\mathbb{Z}^{+}}$ denote the set of all Pauli sequences and let
$F(\Pi^{\mathbb{Z}^{+}})$ denote the set of Pauli sequences with finite support.

\begin{definition}
\label{def:qcc}A rate-$k/n$ quantum convolutional code consists of a basic set
$\mathcal{G}_{0}$\ of $n-k$ generators and all of their $n$-qubit shifts
\cite{PhysRevLett.91.177902,arxiv2004olliv,ieee2007forney}. The generators in
$\mathcal{G}_{0}$ commute with each other and with all of their $n$-qubit
shifts. The parameters $k$ and $n$ satisfy $0\leq k\leq n$ and the basic set
$\mathcal{G}_{0}$ is as follows:%
\[
\mathcal{G}_{0}=\left\{  \mathbf{G}_{i}\in F(\Pi^{\mathbb{Z}^{+}}):1\leq i\leq
n-k\right\}  .
\]
A frame of the code consists of $n$ qubits.
\end{definition}

The operation of a rate-$k/n$ quantum convolutional code begins with the
sender encoding a stream of information qubits. Figure 3 of Ref.
\cite{arx2007wilde} illustrates the basic operation of a quantum convolutional
code. The sender encodes $n-k$ ancilla qubits and $k$ information qubits per
frame \cite{ieee2007grassl,isit2006grassl} and transmits the encoded qubits
over a noisy quantum channel. The above stabilizer $\mathcal{G}_{0}$\ and all
of its shifts act like a parity check matrix for the quantum convolutional
code. The receiver measures the generators in the stabilizer to determine an
error syndrome. It is important that the generators in $\mathcal{G}_{0}$ have
finite weight so that the receiver can perform the measurements and produce an
error syndrome. It is also important that the generators have a block-band
form so that the receiver can perform the measurements online as the noisy
encoded qubits arrive. The receiver processes the error syndrome with a method
such as the Viterbi algorithm \cite{itit1967viterbi}\ or any other decoding
algorithm \cite{book1999conv} to determine the most likely error for each
frame of quantum data. The receiver performs a unitary that reverses the
errors. He finally processes the encoded qubits with a decoding circuit to
recover the original stream of information qubits.

\subsubsection{The P2B Isomorphism}

We now review the P2B isomorphism from the set of phase-free Pauli sequences
to the module over the ring of binary polynomials
\cite{arxiv2004olliv,ieee2007forney,arx2007wilde}. We illustrate it by example
(see Ref. \cite{arx2007wilde}\ for a more rigorous development.)

Suppose the following two basic generators specify a rate-1/3 quantum
convolutional code \cite{isit2005forney,ieee2007forney}:%
\begin{equation}
\cdots\left\vert
\begin{array}
[c]{c}%
III\\
III
\end{array}
\right\vert
\begin{array}
[c]{c}%
XXX\\
ZZZ
\end{array}
\left\vert
\begin{array}
[c]{c}%
XZY\\
ZYX
\end{array}
\right\vert \left.
\begin{array}
[c]{c}%
III\\
III
\end{array}
\right\vert \cdots\label{eq:pauli-gens}%
\end{equation}
The vertical bars indicate that we shift by multiples of three to obtain the
other generators in the quantum convolutional code. Observe that the above two
generators commute with all of their three-qubit shifts.

The P2B isomorphism is a mapping from the above stabilizer generators to a
matrix whose entries are binary polynomials. The left side of the matrix is
the \textquotedblleft Z\textquotedblright\ matrix and the right side of the
matrix is the \textquotedblleft X\textquotedblright\ matrix. We consider the
entries in the first frame of the stabilizer generators in
(\ref{eq:pauli-gens})\ for now and map these entries to a matrix with binary
entries. The first frame of the first generator in (\ref{eq:pauli-gens}) has
\textquotedblleft X\textquotedblright\ entries only and the first frame of the
second generator in (\ref{eq:pauli-gens}) has \textquotedblleft
Z\textquotedblright\ entries only. The binary matrix corresponding to the
entries in the first frame is as follows:%
\[
H_{0}=\left[  \left.
\begin{array}
[c]{ccc}%
0 & 0 & 0\\
1 & 1 & 1
\end{array}
\right\vert
\begin{array}
[c]{ccc}%
1 & 1 & 1\\
0 & 0 & 0
\end{array}
\right]  .
\]
The vertical bar now indicates the separation of the \textquotedblleft
Z\textquotedblright\ matrix on the left and the \textquotedblleft
X\textquotedblright\ matrix on the right. A \textquotedblleft
Y\textquotedblright\ entry maps to a \textquotedblleft1\textquotedblright\ in
both the \textquotedblleft Z\textquotedblright\ and \textquotedblleft
X\textquotedblright\ matrix. Let us consider the entries in the second frame
of (\ref{eq:pauli-gens}). They map to the following binary matrix:%
\[
H_{1}=\left[  \left.
\begin{array}
[c]{ccc}%
0 & 1 & 1\\
1 & 1 & 0
\end{array}
\right\vert
\begin{array}
[c]{ccc}%
1 & 0 & 1\\
0 & 1 & 1
\end{array}
\right]  .
\]
We form a matrix of binary polynomials by incorporating the delay transform or
$D$-transform. The following binary polynomial matrix $H\left(  D\right)  $
fully specifies the quantum convolutional code:%
\begin{align*}
H\left(  D\right)   &  =H_{0}+H_{1}\cdot D\\
&  =\left[  \left.
\begin{array}
[c]{ccc}%
0 & D & D\\
1+D & 1+D & 1
\end{array}
\right\vert
\begin{array}
[c]{ccc}%
1+D & 1 & 1+D\\
0 & D & D
\end{array}
\right]  .
\end{align*}

The above description of a quantum convolutional code with a binary polynomial
matrix is powerful because it allows us to perform manipulations with finite
polynomials rather than with countably-infinite sequences of Pauli matrices
(classical convolutional coding theory exploits the same idea
\cite{book1999conv}). The first and second rows of $H\left(  D\right)  $
capture all of the information about the first and second generators
in\ (\ref{eq:pauli-gens}) and all of their three-qubit shifts. We obtain the
$nl$-shift of either of the above generators by multiplying the corresponding
row in $H\left(  D\right)  $\ by $D^{l}$.

\subsubsection{The Shifted Symplectic Product}

The shifted symplectic product $\odot$ provides a way to determine the
commutative properties of any convolutional stabilizer code
\cite{arxiv2004olliv,arx2007wilde} (See Ref.~\cite{arx2007wilde} for a
detailed discussion of the shifted symplectic product with examples). Let
$z_{1}\left(  D\right)  $ and $z_{2}\left(  D\right)  $ denote the first and
second respective rows of the \textquotedblleft Z\textquotedblright\ matrix of
$H\left(  D\right)  $. Let $x_{1}\left(  D\right)  $ and $x_{2}\left(
D\right)  $ be the first and second respective rows of the \textquotedblleft
X\textquotedblright\ matrix of $H\left(  D\right)  $. Let
\begin{align*}
h_{1}\left(  D\right)   &  =\left(  z_{1}\left(  D\right)  |x_{1}\left(
D\right)  \right)  ,\\
h_{2}\left(  D\right)   &  =\left(  z_{2}\left(  D\right)  |x_{2}\left(
D\right)  \right)  ,
\end{align*}
denote the first and second respective rows of $H\left(  D\right)  $. The
vectors $h_{1}\left(  D\right)  $ and $h_{2}\left(  D\right)  $ specify the
first and second respective generators in (\ref{eq:pauli-gens}). We define the
shifted symplectic product of $h_{1}\left(  D\right)  $ and $h_{2}\left(
D\right)  $ as follows:%
\[
\left(  h_{1}\odot h_{2}\right)  \left(  D\right)  =z_{1}\left(
D^{-1}\right)  \cdot x_{2}\left(  D\right)  +x_{1}\left(  D^{-1}\right)  \cdot
z_{2}\left(  D\right)  ,
\]
where $\cdot$ denotes the binary inner product and addition is binary.

The shifted symplectic product $\left(  h_{1}\odot h_{2}\right)  \left(
D\right)  $ vanishes in the above case. The shifted symplectic products
$\left(  h_{1}\odot h_{1}\right)  \left(  D\right)  $ and $\left(  h_{2}\odot
h_{2}\right)  \left(  D\right)  $ also vanish. The shifted symplectic product
between two vectors of binary polynomials vanishes if and only if their
corresponding Pauli sequences commute \cite{arxiv2004olliv,arx2007wilde}. Time
reversal (substituting $D^{-1}$ for $D$) ensures that the shifted symplectic
product checks commutativity for every shift of the two Pauli sequences being
compared. The cases where the shifted symplectic product does not vanish
(where the two Pauli sequences anticommute for one or more shifts) are
important for constructing entanglement-assisted quantum convolutional codes.

\subsubsection{Row and Column Operations}

We can perform row operations on binary polynomial matrices for quantum
convolutional codes. A row operation is merely a \textquotedblleft
mental\textquotedblright\ operation that has no effect on the states in the
codespace or on the error-correcting properties of the code. We have three
types of row operations:

\begin{enumerate}
\item An elementary row operation multiplies a row times an arbitrary binary
polynomial and adds the result to another row. This additive invariance holds
for any code that admits a description within the stabilizer formalism.
Additive codes are invariant under multiplication of the stabilizer generators
in the \textquotedblleft Pauli\ picture\textquotedblright\ or under row
addition in the \textquotedblleft
binary-polynomial\ picture.\textquotedblright

\item Another type of row operation is to multiply a row by an arbitrary power
of $D$. Ollivier and Tillich discuss such row operations as \textquotedblleft
multiplication of a line by $D$\textquotedblright\ and use them to find
encoding operations for their quantum convolutional codes
\cite{arxiv2004olliv}. Grassl and R\"{o}tteler use this type of operation to
find a subcode of a given quantum convolutional code with an equivalent
asymptotic rate and equivalent error-correcting properties
\cite{isit2006grassl}. We use this type of row operation in each of our three
classes of entanglement-assisted quantum convolutional codes.

\item We also employ row operations that multiply a row by an arbitrary
polynomial (not necessarily a power of $D$). We only use these operations when
the receiver performs a measurement to diagnose an error. This type of row
operation occurs when we have generators with infinite weight that we would
like to reduce to finite weight so that the receiver can perform measurements
in an online fashion as qubits arrive from the noisy channel. We use this type
of row operation in our second and third classes of entanglement-assisted
quantum convolutional codes.
\end{enumerate}

A\ row operation does not change the shifted symplectic product when all
generators commute. A row operation \textit{does} change the shifted
symplectic product of a set of generators that do not commute. It is a
convenient tool for constructing our entanglement-assisted quantum
convolutional codes.

We can also perform column operations on binary polynomial matrices for
quantum convolutional codes. Column operations change the error-correcting
properties of the code and are important for realizing a periodic encoding
circuit for the code. We have two types of column operations:

\begin{enumerate}
\item An elementary column operation multiplies one column by an arbitrary
binary polynomial and adds the result to another column. We implement
elementary column operations with gates from the shift-invariant Clifford
group \cite{ieee2007grassl,isit2006grassl}.

\item Another column operation is to multiply column $i$\ in both the
\textquotedblleft X\textquotedblright\ and \textquotedblleft
Z\textquotedblright\ matrix by $D^{l}$ where$\ l\in\mathbb{Z}$. We perform
this operation by delaying or advancing the processing of qubit $i$\ by $l$
frames relative to the original frame.
\end{enumerate}

A column operation implemented on the \textquotedblleft X\textquotedblright%
\ side of the binary polynomial matrix has a corresponding effect on the
\textquotedblleft Z\textquotedblright\ side of the binary polynomial matrix.
This corresponding effect is a manifestation of the Heisenberg uncertainty
principle because commutation relations remain invariant with respect to the
action of quantum gates. The shifted symplectic product is therefore invariant
with respect to column operations from the shift-invariant Clifford group. We
describe possible column operations for implementing encoding circuits in the
next section.

\section{Finite-Depth Clifford Operations}

One of the main advantages of a quantum convolutional code is that its
encoding circuit has a periodic form. We can encode a stream of quantum
information with the same physical routines or devices and therefore reduce
encoding and decoding complexity.

Ollivier and Tillich were the first to discuss encoding circuits for quantum
convolutional codes \cite{PhysRevLett.91.177902,arxiv2004olliv}. They provided
a set of necessary and sufficient conditions to determine when an encoding
circuit is noncatastrophic. A noncatastrophic encoding circuit does not
propagate uncorrected errors infinitely through the decoded information qubit
stream. Classical convolutional coding theory has a well developed theory of
noncatastrophic encoding circuits \cite{book1999conv}.

Grassl and R\"{o}tteler later showed that Ollivier and Tillich's conditions
for a circuit to be noncatastrophic are too restrictive
\cite{isit2006grassl,ieee2006grassl,ieee2007grassl}. They found subcodes of
quantum convolutional codes that admit noncatastrophic encoders. The
noncatastrophic encoders are a sequence of Clifford circuits with finite
depth. They developed a formalism for encapsulating the periodic structure of
an encoding circuit by decomposing the encoding circuit as a set of elementary
column operations. Their decoding circuits are exact inverses of their
encoding circuits because their decoding circuits perform the encoding
operations in reverse order.

\begin{definition}
A \textit{finite-depth operation} transforms every finite-weight\ stabilizer
generator to one with finite weight.
\end{definition}

\label{sec:finite-depth-ops}We review the finite-depth operations in the
shift-invariant Clifford group
\cite{isit2006grassl,ieee2006grassl,ieee2007grassl}. The shift-invariant
Clifford group is an extension of the Clifford group operations mentioned in
Section~\ref{sec:stabilizer-review}. We describe how finite-depth operations
in the shift-invariant Clifford group affect the binary polynomial matrix for
a code. Each of the following operations acts on every frame of a quantum
convolutional code.

\begin{enumerate}
\item The sender performs a CNOT\ from qubit $i$ to qubit $j$ in every frame
where qubit $j$ is in a frame delayed by $k$. The effect on the binary
polynomial matrix is to multiply column $i$ by $D^{k}$ and add the result to
column $j$ in the \textquotedblleft X\textquotedblright\ matrix and to
multiply column$\ j$ by $D^{-k}$ and add the result to column $i$ in the
\textquotedblleft Z\textquotedblright\ matrix.

\item A Hadamard on qubit $i$ swaps column $i$ in the \textquotedblleft
X\textquotedblright\ matrix with column $i$ in the \textquotedblleft
Z\textquotedblright\ matrix.

\item A phase gate on qubit $i$ adds column $i$ from the \textquotedblleft
X\textquotedblright\ matrix to column $i$ in the \textquotedblleft
Z\textquotedblright\ matrix.

\item A controlled-phase gate from qubit $i$ to qubit $j$ in a frame delayed
by $k$ multiplies column $i$ in the \textquotedblleft X\textquotedblright%
\ matrix by $D^{k}$ and adds the result to column $j$ in the \textquotedblleft
Z\textquotedblright\ matrix. It also multiples column $j$ in the
\textquotedblleft X\textquotedblright\ matrix by $D^{-k}$ and adds the result
to column $i$ in the \textquotedblleft Z\textquotedblright\ matrix.

\item A controlled-phase gate from qubit $i$ to qubit $i$ in a frame delayed
by $k$ multiplies column $i$ in the \textquotedblleft X\textquotedblright%
\ matrix by $D^{k}+D^{-k}$ and adds the result to column $i$ in the
\textquotedblleft Z\textquotedblright\ matrix.
\end{enumerate}

We use finite-depth operations extensively in this work, but we employ only
the above Hadamard and CNOT\ gates because our entanglement-assisted quantum
convolutional codes have the CSS\ structure.
Figure~\ref{fig:example-eaqcc-fefd}\ gives an example of an
entanglement-assisted quantum convolutional code that employs several
finite-depth operations. The circuit encodes a stream of information qubits
with the help of ebits shared between sender and receiver.

Multiple CNOT\ gates can realize an elementary column operation as described
at the end of Section~\ref{sec:rev-conv-stab}. Suppose the elementary column
operation multiplies column $i$ in the \textquotedblleft X\textquotedblright%
\ matrix by $f\left(  D\right)  $ and adds the result to column $j$.
Polynomial $f\left(  D\right)  $ is a summation of some finite set $\left\{
l_{1},\ldots,l_{n}\right\}  $ of powers of $D$:%
\[
f\left(  D\right)  =D^{l_{1}}+\cdots+D^{l_{n}}.
\]
We can realize $f\left(  D\right)  $ by performing a CNOT\ gate from qubit $i$
to qubit $j$ in a frame delayed by $l_{i}$ for each $i\in\left\{
1,\ldots,n\right\}  $.

\section{CSS\ Entanglement-Assisted Quantum Convolutional Codes}

\label{sec:eaqcc}%
\begin{figure*}
[ptb]
\begin{center}
\includegraphics[
natheight=8.879900in,
natwidth=19.139999in,
height=2.9573in,
width=6.141in
]
{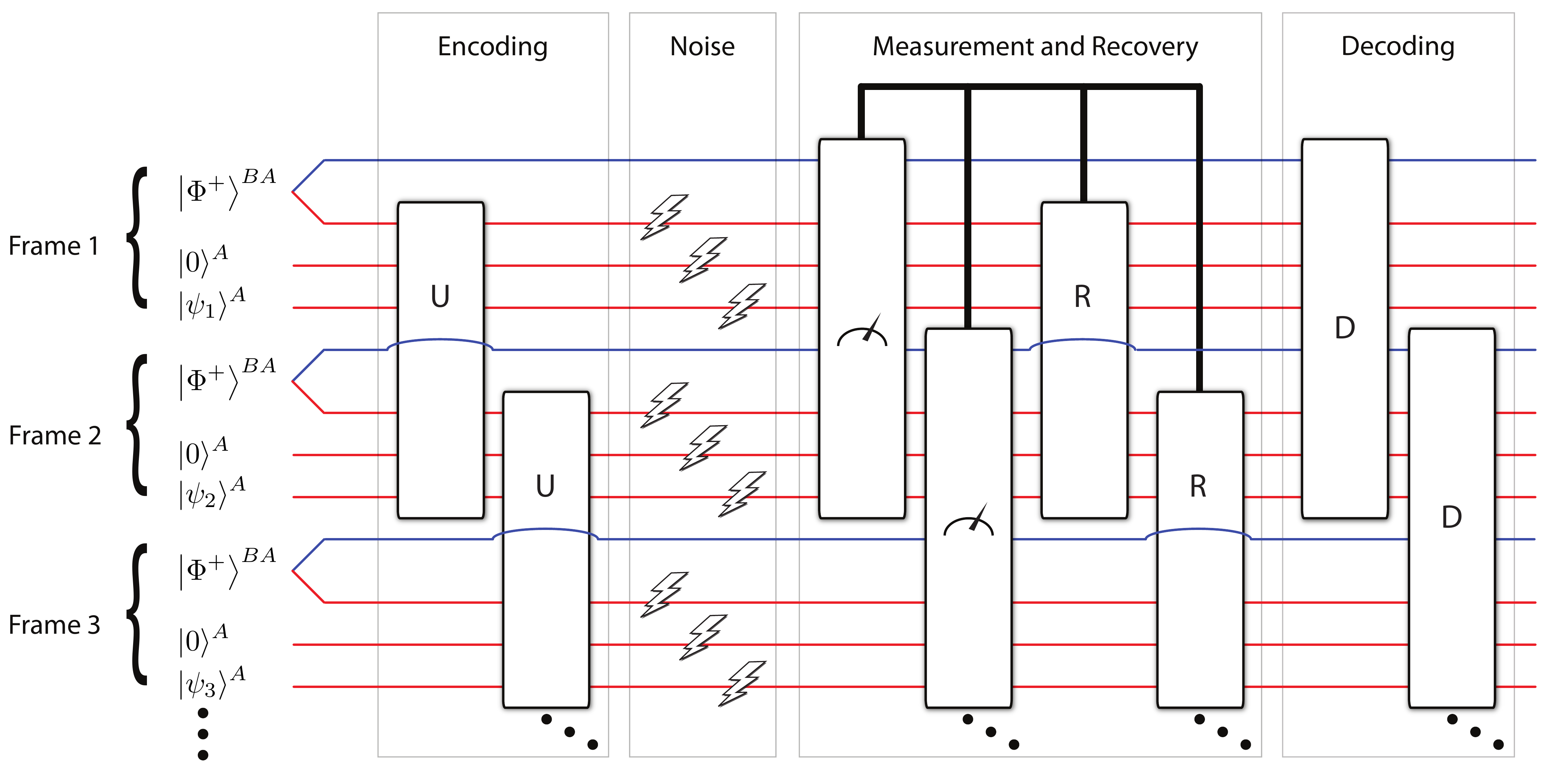}
\end{center}
\caption
{(Color online) An entanglement-assisted quantum convolutional code operates on a stream of qubits
partitioned into a countable number of frames. The sender encodes
the frames of information qubits, ancilla qubits, and half of shared ebits with
a repeated, overlapping encoding circuit $U$. The noisy channel affects the sender's encoded qubits
but does not affect the receiver's half of the shared ebits. The receiver performs overlapping
measurements on both the encoded qubits and his half of the shared ebits. These measurements
produce an error syndrome which the receiver can process to determine the most likely error.
The receiver reverses the errors on the noisy qubits from the sender. The final decoding circuit operates
on all qubits in a frame and recovers
the original stream of information qubits.}
\label{fig:eaqcc}
\end{figure*}%
An entanglement-assisted quantum convolutional code operates similarly to a
standard quantum convolutional code. The main difference is that the sender
and receiver share entanglement in the form of ebits. An $[[n,k;c]]$
entanglement-assisted quantum convolutional code encodes $k$ information
qubits per frame with the help of $c$ ebits and $n-k-c$ ancilla qubits per
frame. Figure~\ref{fig:eaqcc} highlights the main features of the operation of
an entanglement-assisted quantum convolutional code. The sender encodes a
stream of quantum information using both additional ancillas and ebits. The
sender performs the encoding operations on her qubits only (i.e., not
including the halves of the ebits in possession of the receiver). The encoding
operations have a periodic structure so that the same operations act on qubits
in different frames and give the code a memory structure. The sender can
perform these encoding operations in an online manner as she places more
qubits in the unencoded qubit stream. The sender transmits her encoded qubits
over the noisy quantum communication channel. The noisy channel does not
affect the receiver's half of the shared ebits. The receiver combines the
received noisy qubits with his half of the ebits and performs measurements to
diagnose errors that may occur. These measurements may overlap on some of the
same qubits. The receiver then diagnoses errors using a classical technique
such as Viterbi error estimation \cite{itit1967viterbi}, reverses the errors
that the channel introduces, and finally performs an online decoding circuit
that outputs the original information qubit stream.

Our main theorem below allows us to import two arbitrary classical
convolutional codes for use as a CSS entanglement-assisted quantum
convolutional code. Grassl and R\"{o}tteler were the first to construct
CSS\ quantum convolutional codes from two classical binary convolutional codes
that satisfy an orthogonality constraint---the polynomial parity check
matrices $H_{1}\left(  D\right)  $ and $H_{2}\left(  D\right)  $ of the two
classical codes are orthogonal with respect to the shifted symplectic product
\cite{ieee2007grassl}:%
\begin{equation}
H_{1}\left(  D\right)  H_{2}^{T}\left(  D^{-1}\right)  =0.
\end{equation}
The resulting symplectic code has a self-orthogonal parity-check matrix when
we join them together using the CSS\ construction. Our theorem generalizes the
work of Grassl and R\"{o}tteler because we can import two \textit{arbitrary}
classical binary convolutional codes---the codes do not necessarily have to
obey the self-orthogonality constraint.

The theorem gives a direct way to compute the amount of entanglement that the
code requires. The number of ebits required is equal to the rank of a
particular matrix derived from the check matrices of the two classical codes.
It generalizes an earlier theorem that determines the amount of entanglement
required for an entanglement-assisted quantum block code \cite{isit2007brun}.

Theorem~\ref{thm:main} also provides a formula to compute the performance
parameters of the entanglement-assisted quantum convolutional code from the
performance parameters of the two classical codes. This formula ensures that
high-rate classical convolutional codes produce high-rate
entanglement-assisted quantum convolutional codes. Our constructions also
ensure high performance for the \textquotedblleft trade-off\textquotedblright%
\ and \textquotedblleft catalytic\textquotedblright\ rates by minimizing the
number of ebits that the codes require.

We begin the proof of the theorem in this section and complete it in different
ways for each of our two classes of entanglement-assisted quantum
convolutional codes in Sections~\ref{sec:eaqcc-fefd} and \ref{sec:eaqcc-iefd}.
The proofs detail how to encode a stream of information qubits, ancilla
qubits, and shared ebits into a code that has the CSS\ structure.

\begin{theorem}
\label{thm:main}Let $H_{1}\left(  D\right)  $ and $H_{2}\left(  D\right)  $ be
the respective check matrices corresponding to noncatastrophic, delay-free
encoders for classical binary convolutional codes $C_{1}$ and $C_{2}$. Suppose
that classical code $C_{i}$ encodes $k_{i}$ information bits with $n$ bits per
frame where $i=1,2$. The respective\ dimensions of $H_{1}\left(  D\right)  $
and $H_{2}\left(  D\right)  $ are thus $\left(  n-k_{1}\right)  \times n$ and
$\left(  n-k_{2}\right)  \times n$. Then the resulting entanglement-assisted
quantum convolutional code encodes $k_{1}+k_{2}-n+c$ information qubits per
frame and is an $\left[  \left[  n,k_{1}+k_{2}-n+c;c\right]  \right]  $
entanglement-assisted quantum convolutional code. The code requires $c$ ebits
per frame where $c$ is equal to the rank of $H_{1}\left(  D\right)  H_{2}%
^{T}\left(  D^{-1}\right)  $.
\end{theorem}

Let us begin the proof of the above theorem by constructing an
entanglement-assisted quantum convolutional code. Consider the following
quantum check matrix in CSS\ form:%
\begin{equation}
\left[  \left.
\begin{array}
[c]{c}%
H_{1}\left(  D\right) \\
0
\end{array}
\right\vert
\begin{array}
[c]{c}%
0\\
H_{2}\left(  D\right)
\end{array}
\right]  . \label{eq:orig-gens}%
\end{equation}
We label the above matrix as a \textquotedblleft quantum check
matrix\textquotedblright\ for now because it does not necessarily correspond
to a commuting stabilizer. The quantum check matrix corresponds to a set of
Pauli sequences whose error-correcting properties are desirable.

The following lemma begins the proof of the above theorem. It details an
initial decomposition of the above quantum check matrix for each of our two
classes of entanglement-assisted quantum convolutional codes.

\begin{lemma}
\label{lemma:general-ops}Elementary row and column operations relate the
quantum check matrix\ in (\ref{eq:orig-gens}) to the following matrix%
\[
\left[  \left.
\begin{array}
[c]{cc}%
E\left(  D\right)  & F\left(  D\right) \\
0 & 0
\end{array}
\right\vert
\begin{array}
[c]{cc}%
0 & 0\\
I & 0
\end{array}
\right]  .
\]
where $E\left(  D\right)  $ is dimension $\left(  n-k_{1}\right)
\times\left(  n-k_{2}\right)  $, $F\left(  D\right)  $ is $\left(
n-k_{1}\right)  \times k_{2}$, the identity matrix is $\left(  n-k_{2}\right)
\times\left(  n-k_{2}\right)  $, and the null matrix on the right is $\left(
n-k_{2}\right)  \times k_{2}$. We give a definition of $E\left(  D\right)  $
and $F\left(  D\right)  $ in the following proof.
\end{lemma}

%

\begin{proof}%
The Smith form \cite{book1999conv} of $H_{i}\left(  D\right)  $ for each
$i=1,2$\ is%
\begin{equation}
H_{i}\left(  D\right)  =A_{i}\left(  D\right)  \left[
\begin{array}
[c]{cc}%
I & 0
\end{array}
\right]  B_{i}\left(  D\right)  , \label{eq:first-Smith-form}%
\end{equation}
where $A_{i}\left(  D\right)  $ is $\left(  n-k_{i}\right)  \times\left(
n-k_{i}\right)  $, the matrix in brackets is $\left(  n-k_{i}\right)  \times
n$, and $B_{i}\left(  D\right)  $ is $n\times n$ \cite{book1999conv}. Let
$B_{ia}\left(  D\right)  $\ be the first $n-k_{i}$ rows of $B_{i}\left(
D\right)  $ and let $B_{ib}\left(  D\right)  $ be the last $k_{i}$ rows of
$B_{i}\left(  D\right)  $:%
\[
B_{i}\left(  D\right)  =\left[
\begin{array}
[c]{c}%
B_{ia}\left(  D\right) \\
B_{ib}\left(  D\right)
\end{array}
\right]  .
\]
The $\left(  n-k_{i}\right)  \times\left(  n-k_{i}\right)  $ identity matrix
in brackets in (\ref{eq:first-Smith-form}) indicates that the invariant
factors of $H_{i}\left(  D\right)  $ for each $i=1,2$ are all equal to one
\cite{book1999conv}. The invariant factors are all unity for both check
matrices because the check matrices correspond to noncatastrophic, delay-free
encoders \cite{book1999conv}. The matrices $A_{i}\left(  D\right)  $ and
$B_{i}\left(  D\right)  $ are a product of a sequence of elementary row and
column operations respectively \cite{book1999conv}.

Premultiplying $H_{i}\left(  D\right)  $ by $A_{i}^{-1}\left(  D\right)  $
gives a check matrix $H_{i}^{\prime}\left(  D\right)  $ for each $i=1,2$.
Matrix $H_{i}^{\prime}\left(  D\right)  $ is a check matrix for code $C_{i}$
with equivalent error-correcting properties as $H_{i}\left(  D\right)
$\ because row operations relate the two matrices. This new check matrix
$H_{i}^{\prime}\left(  D\right)  $\ is equal to the first $n-k_{i}$ rows of
matrix $B_{i}\left(  D\right)  $:%
\[
H_{i}^{\prime}\left(  D\right)  =B_{ia}\left(  D\right)  .
\]

The invariant factors of $H_{1}\left(  D\right)  H_{2}^{T}\left(
D^{-1}\right)  $ are equivalent to those of $H_{1}^{\prime}\left(  D\right)
H_{2}^{\prime T}\left(  D^{-1}\right)  $ because they are related by row and
column operations \cite{book1999conv}:%
\begin{multline}
H_{1}\left(  D\right)  H_{2}^{T}\left(  D^{-1}\right)
=\label{eq:check-relation}\\
A_{1}\left(  D\right)  H_{1}^{\prime}\left(  D\right)  H_{2}^{\prime T}\left(
D^{-1}\right)  A_{2}^{T}\left(  D^{-1}\right)  .
\end{multline}

We now decompose the above quantum check matrix into a basic form using
elementary row and column operations. The row operations have no effect on the
error-correcting properties of the code, and the column operations correspond
to elements of an encoding circuit. We later show how to incorporate ebits so
that the quantum check matrix forms a valid commuting stabilizer.

Perform the row operations in matrices $A_{i}^{-1}\left(  D\right)  $ for both
check matrices $H_{i}\left(  D\right)  $. The quantum check matrix becomes%
\begin{equation}
\left[  \left.
\begin{array}
[c]{c}%
B_{1a}\left(  D\right) \\
0
\end{array}
\right\vert
\begin{array}
[c]{c}%
0\\
B_{2a}\left(  D\right)
\end{array}
\right]  . \label{eq:orig-qcm}%
\end{equation}
The error-correcting properties of the above generators are equivalent to
those of the generators in (\ref{eq:orig-gens}) because row operations relate
the two sets of generators. The matrix $B_{2}\left(  D\right)  $ corresponds
to a sequence of elementary column operations $B_{2,i}\left(  D\right)  $:%
\[
B_{2}\left(  D\right)  =B_{2,1}\left(  D\right)  \cdots B_{2,l}\left(
D\right)  =\prod_{i=1}^{l}B_{2,i}\left(  D\right)  .
\]
The inverse matrix $B_{2}^{-1}\left(  D\right)  $ is therefore equal to the
above sequence of operations in reverse order:%
\[
B_{2}^{-1}\left(  D\right)  =B_{2,l}\left(  D\right)  \cdots B_{2,1}\left(
D\right)  =\prod_{i=l}^{1}B_{2,i}\left(  D\right)  .
\]
Perform the elementary column operations in $B_{2}^{-1}\left(  D\right)  $
with CNOT\ and SWAP\ gates \cite{isit2006grassl}. The effect of each
elementary column operation $B_{2,i}\left(  D\right)  $ is to postmultiply the
\textquotedblleft X\textquotedblright\ matrix by $B_{2,i}\left(  D\right)  $
and to postmultiply the \textquotedblleft Z\textquotedblright\ matrix by
$B_{2,i}^{T}\left(  D^{-1}\right)  $. Therefore the effect of all elementary
operations is to postmultiply the \textquotedblleft Z\textquotedblright%
\ matrix by $B_{2}^{T}\left(  D^{-1}\right)  $ because%
\[
\prod_{i=l}^{1}B_{2,i}^{T}\left(  D^{-1}\right)  =\left(  \prod_{i=1}%
^{l}B_{2,i}\left(  D^{-1}\right)  \right)  ^{T}=B_{2}^{T}\left(
D^{-1}\right)  .
\]
The quantum check matrix in (\ref{eq:orig-qcm})\ becomes%
\begin{equation}
\left[  \left.
\begin{array}
[c]{c}%
B_{1a}\left(  D\right)  B_{2}^{T}\left(  D^{-1}\right) \\
0
\end{array}
\right\vert
\begin{array}
[c]{cc}%
0 & 0\\
I & 0
\end{array}
\right]  . \label{eq:lemma-temp-QCM}%
\end{equation}
Let $E\left(  D\right)  $ be equal to the first $n-k_{1}$ rows and $n-k_{2}$
columns of the \textquotedblleft Z\textquotedblright\ matrix:%
\[
E\left(  D\right)  \equiv B_{1,a}\left(  D\right)  B_{2,a}^{T}\left(
D^{-1}\right)  ,
\]
and let $F\left(  D\right)  $ be equal to the first $n-k_{1}$ rows and last
$k_{2}$ columns of the \textquotedblleft Z\textquotedblright\ matrix:%
\[
F\left(  D\right)  \equiv B_{1,a}\left(  D\right)  B_{2,b}^{T}\left(
D^{-1}\right)  .
\]
The quantum check matrix in (\ref{eq:lemma-temp-QCM}) is then equivalent to
the following matrix%
\begin{equation}
\left[  \left.
\begin{array}
[c]{cc}%
E\left(  D\right)  & F\left(  D\right) \\
0 & 0
\end{array}
\right\vert
\begin{array}
[c]{cc}%
0 & 0\\
I & 0
\end{array}
\right]  , \label{eq:last-eq-fefd}%
\end{equation}
where each matrix above has the dimensions stated in the theorem above.%
\end{proof}%
The above operations end the initial set of operations that each of our two
classes of entanglement-assisted quantum convolutional codes employs. We
outline the remaining operations for each class of codes in what follows.

\section{Entanglement-Assisted Quantum Convolutional Codes with Finite-Depth
Encoding and Decoding Circuits}

\label{sec:eaqcc-fefd}This section details entanglement-assisted quantum
convolutional codes in our first class. Codes in the first class admit an
encoding and decoding circuit that employ finite-depth operations only. The
check matrices for codes in this class have a property that allows this type
of encoding and decoding. The following lemma gives the details of this
property, and the proof outlines how to encode and decode this class of
entanglement-assisted quantum convolutional codes.

\begin{lemma}
\label{lemma:fefd}Suppose the Smith form\ of $H_{1}\left(  D\right)  H_{2}%
^{T}\left(  D^{-1}\right)  $ is%
\[
H_{1}\left(  D\right)  H_{2}^{T}\left(  D^{-1}\right)  =A\left(  D\right)
\left[
\begin{array}
[c]{cc}%
\Gamma\left(  D\right)  & 0\\
0 & 0
\end{array}
\right]  B\left(  D\right)  ,
\]
where $A\left(  D\right)  $ is an $\left(  n-k_{1}\right)  \times\left(
n-k_{1}\right)  $ matrix, $B\left(  D\right)  $ is an $\left(  n-k_{2}\right)
\times\left(  n-k_{2}\right)  $ matrix, $\Gamma\left(  D\right)  $ is a
diagonal $c\times c$ matrix whose entries are powers of $D$, and the matrix in
brackets has dimension $\left(  n-k_{1}\right)  \times\left(  n-k_{2}\right)
$. Then the resulting entanglement-assisted quantum convolutional code has
both a finite-depth encoding and decoding circuit.
\end{lemma}

%

\begin{proof}%
We begin the proof of this theorem by continuing where the proof of
Lemma~\ref{lemma:general-ops} ends. The crucial assumption for the above lemma
is that the invariant factors of $H_{1}\left(  D\right)  H_{2}^{T}\left(
D^{-1}\right)  $ are all powers of $D$. The Smith form of $E\left(  D\right)
$ in (\ref{eq:last-eq-fefd}) therefore becomes%
\[
A_{1}^{-1}\left(  D\right)  A\left(  D\right)  \left[
\begin{array}
[c]{cc}%
\Gamma\left(  D\right)  & 0\\
0 & 0
\end{array}
\right]  B\left(  D\right)  A_{2}^{-1}\left(  D\right)  ,
\]
by employing the hypothesis of Lemma~\ref{lemma:fefd} and
(\ref{eq:check-relation}). The rank of both $H_{1}\left(  D\right)  H_{2}%
^{T}\left(  D^{-1}\right)  $ and $E\left(  D\right)  $ is equal to $c$.

Perform the inverse of the row operations in $A_{1}^{-1}\left(  D\right)
A\left(  D\right)  $ on the first $n-k_{1}$ rows of the quantum check matrix
in (\ref{eq:last-eq-fefd}). Perform the inverse of the column operations in
matrix $B\left(  D\right)  A_{2}^{-1}\left(  D\right)  $ on the first
$n-k_{2}$ columns of the quantum check matrix in (\ref{eq:last-eq-fefd}). We
execute these column operations with Hadamard, CNOT,\ and SWAP\ gates. These
column operations have a corresponding effect on columns in the
\textquotedblleft X\textquotedblright\ matrix, but we can exploit the identity
matrix in the last $n-k_{2}$ rows of the \textquotedblleft X\textquotedblright%
\ matrix to counteract this effect. We perform row operations on the last
$n-k_{2}$ rows of the matrix that act as the inverse of the column operations,
and therefore the quantum check matrix in (\ref{eq:last-eq-fefd})\ becomes%
\[
\left[  \left.
\begin{array}
[c]{ccc}%
\Gamma\left(  D\right)  & 0 & F_{1}\left(  D\right) \\
0 & 0 & F_{2}\left(  D\right) \\
0 & 0 & 0\\
0 & 0 & 0
\end{array}
\right\vert
\begin{array}
[c]{ccc}%
0 & 0 & 0\\
0 & 0 & 0\\
I & 0 & 0\\
0 & I & 0
\end{array}
\right]  ,
\]
where $F_{1}\left(  D\right)  $ and $F_{2}\left(  D\right)  $ are the first
$c$ and $n-k_{1}-c$\ respective rows of $A^{-1}\left(  D\right)  A_{1}\left(
D\right)  F\left(  D\right)  $. We perform Hadamard and CNOT\ gates to clear
the entries in $F_{1}\left(  D\right)  $ in the \textquotedblleft
Z\textquotedblright\ matrix above. The quantum check matrix becomes%
\begin{equation}
\left[  \left.
\begin{array}
[c]{ccc}%
\Gamma\left(  D\right)  & 0 & 0\\
0 & 0 & F_{2}\left(  D\right) \\
0 & 0 & 0\\
0 & I & 0
\end{array}
\right\vert
\begin{array}
[c]{ccc}%
0 & 0 & 0\\
0 & 0 & 0\\
I & 0 & 0\\
0 & 0 & 0
\end{array}
\right]  . \label{eq:init-example-inter-step}%
\end{equation}

The Smith form of $F_{2}\left(  D\right)  $ is%
\[
F_{2}\left(  D\right)  =A_{F}\left(  D\right)  \left[
\begin{array}
[c]{cc}%
\Gamma_{F}\left(  D\right)  & 0
\end{array}
\right]  B_{F}\left(  D\right)  ,
\]
where $\Gamma_{F}\left(  D\right)  $ is a diagonal matrix whose entries are
powers of $D$, $A_{F}\left(  D\right)  $ is $\left(  n-k_{1}-c\right)
\times\left(  n-k_{1}-c\right)  $, and $B_{F}\left(  D\right)  $ is
$k_{2}\times k_{2}$. The Smith form of $F_{2}\left(  D\right)  $ takes this
particular form because the original check matrix $H_{2}\left(  D\right)  $ is
noncatastrophic and column operations with Laurent polynomials change the
invariant factors only up to powers of $D$.

Perform row operations corresponding to $A_{F}^{-1}\left(  D\right)  $ on the
second set of $n-k_{1}-c$ rows with $F_{2}\left(  D\right)  $ in
(\ref{eq:init-example-inter-step}). Perform column operations corresponding to
$B_{F}^{-1}\left(  D\right)  $ on columns $n-k_{2}+1,\ldots,n$ with Hadamard,
CNOT, and SWAP\ gates. The resulting quantum check matrix has the following
form:%
\begin{equation}
\left[  \left.
\begin{array}
[c]{cccc}%
\Gamma\left(  D\right)  & 0 & 0 & 0\\
0 & 0 & \Gamma_{F}\left(  D\right)  & 0\\
0 & 0 & 0 & 0\\
0 & I & 0 & 0
\end{array}
\right\vert
\begin{array}
[c]{cccc}%
0 & 0 & 0 & 0\\
0 & 0 & 0 & 0\\
I & 0 & 0 & 0\\
0 & 0 & 0 & 0
\end{array}
\right]  . \label{eq:final-quantum-check-matrix}%
\end{equation}

We have now completed the decomposition of the original quantum check matrix
in (\ref{eq:orig-gens}) for this class of entanglement-assisted quantum
convolutional codes. It is not possible to perform row or column operations to
decompose the above matrix any further. The problem with the above quantum
check matrix is that it does not form a valid quantum convolutional code. The
first set of rows with matrix $\Gamma\left(  D\right)  $ are not orthogonal
under the shifted symplectic product to the third set of rows with the
identity matrix on the \textquotedblleft X\textquotedblright\ side.
Equivalently, the set of Pauli sequences corresponding to the above quantum
check matrix do not form a commuting stabilizer. We can use entanglement
shared between sender and receiver to solve this problem. Entanglement adds
columns to the above quantum check matrix to resolve the issue. The additional
columns correspond to qubits on the receiver's side. We next show in detail
how to incorporate ancilla qubits, ebits, and information qubits to obtain a
valid stabilizer code. The result is that we can exploit the error-correcting
properties of the original code to protect the sender's qubits.

Consider the following check matrix corresponding to a commuting stabilizer:%
\begin{equation}
\left[  \left.
\begin{array}
[c]{ccccc}%
I & I & 0 & 0 & 0\\
0 & 0 & 0 & I & 0\\
0 & 0 & 0 & 0 & 0\\
0 & 0 & I & 0 & 0
\end{array}
\right\vert
\begin{array}
[c]{ccccc}%
0 & 0 & 0 & 0 & 0\\
0 & 0 & 0 & 0 & 0\\
I & I & 0 & 0 & 0\\
0 & 0 & 0 & 0 & 0
\end{array}
\right]  , \label{eq:bare-1st-stab}%
\end{equation}
where the identity matrices in the first and third sets of rows each have
dimension $c\times c$, the identity matrix in the second set of rows has
dimension $\left(  n-k_{1}-c\right)  \times\left(  n-k_{1}-c\right)  $, and
the identity matrix in the fourth set of rows has dimension $\left(
n-k_{2}-c\right)  \times\left(  n-k_{2}-c\right)  $. The first and third sets
of $c$\ rows stabilize a set of $c$ ebits shared between Alice and Bob. Bob
possesses the \textquotedblleft left\textquotedblright\ $c$ qubits and Alice
possesses the \textquotedblleft right\textquotedblright\ $n$ qubits. The
second and fourth sets of rows stabilize a set of $2\left(  n-c\right)
-k_{1}-k_{2}$ ancilla qubits that Alice possesses. The stabilizer therefore
stabilizes a set of $c$ ebits, $2\left(  n-c\right)  -k_{1}-k_{2}$ ancilla
qubits, and $k_{1}+k_{2}-n+c$ information qubits.

Observe that the last $n$ columns of the \textquotedblleft Z\textquotedblright%
\ and \textquotedblleft X\textquotedblright\ matrices in the above stabilizer
are similar in their layout to the entries in
(\ref{eq:final-quantum-check-matrix}). We can delay the rows of the above
stabilizer by an arbitrary amount to obtain the desired stabilizer. So the
above stabilizer is a subcode of the following stabilizer in the sense of
Ref.~\cite{isit2006grassl}:%
\[
\left[  \left.
\begin{array}
[c]{ccccc}%
\Gamma\left(  D\right)  & \Gamma\left(  D\right)  & 0 & 0 & 0\\
0 & 0 & 0 & \Gamma_{F}\left(  D\right)  & 0\\
0 & 0 & 0 & 0 & 0\\
0 & 0 & I & 0 & 0
\end{array}
\right\vert
\begin{array}
[c]{ccccc}%
0 & 0 & 0 & 0 & 0\\
0 & 0 & 0 & 0 & 0\\
I & I & 0 & 0 & 0\\
0 & 0 & 0 & 0 & 0
\end{array}
\right]  .
\]
The stabilizer in (\ref{eq:bare-1st-stab}) has equivalent error-correcting
properties to and the same asymptotic rate as the above desired stabilizer.
The above stabilizer matrix is an augmented version of the quantum check
matrix in (\ref{eq:final-quantum-check-matrix}) that includes entanglement.
The sender performs all of the encoding column operations detailed in the
proofs of this lemma and Lemma~\ref{lemma:general-ops} in reverse order. The
result of these operations is an $\left[  \left[  n,k_{1}+k_{2}-n+c;c\right]
\right]  $ entanglement-assisted quantum convolutional code with the same
error-correcting properties as the quantum check matrix in (\ref{eq:orig-gens}%
). The receiver decodes the original information-qubit stream by performing
the column operations in the order presented. The information qubits appear as
the last $k_{1}+k_{2}-n+c$ in each frame of the stream (corresponding to the
$k_{1}+k_{2}-n+c$ columns of zeros in both the \textquotedblleft
Z\textquotedblright\ and \textquotedblleft X\textquotedblright\ matrices
above).%
\end{proof}%

\begin{example}
\label{ex:fefd-example}Consider a classical convolutional code with the
following check matrix:%
\[
H\left(  D\right)  =\left[
\begin{array}
[c]{cc}%
1+D^{2} & 1+D+D^{2}%
\end{array}
\right]  .
\]
We can use $H\left(  D\right)  $ in an entanglement-assisted quantum
convolutional code to correct for both bit-flip errors and phase-flip errors.
We form the following quantum check matrix:%
\begin{equation}
\left[  \left.
\begin{array}
[c]{cc}%
1+D^{2} & 1+D+D^{2}\\
0 & 0
\end{array}
\right\vert
\begin{array}
[c]{cc}%
0 & 0\\
1+D^{2} & 1+D+D^{2}%
\end{array}
\right]  . \label{eq:desired-QCM-first-example}%
\end{equation}
This code falls in the first class of entanglement-assisted quantum
convolutional codes because $H\left(  D\right)  H^{T}\left(  D^{-1}\right)
=1$.\newline\newline We do not show the decomposition of the above check
matrix as outlined in Lemma~\ref{lemma:fefd}, but instead show how to encode
it starting from a stream of information qubits and ebits. Each frame has one
ebit and one information qubit.%
\begin{figure}
[ptb]
\begin{center}
\includegraphics[
natheight=9.753400in,
natwidth=7.920000in,
height=4.0335in,
width=3.2802in
]%
{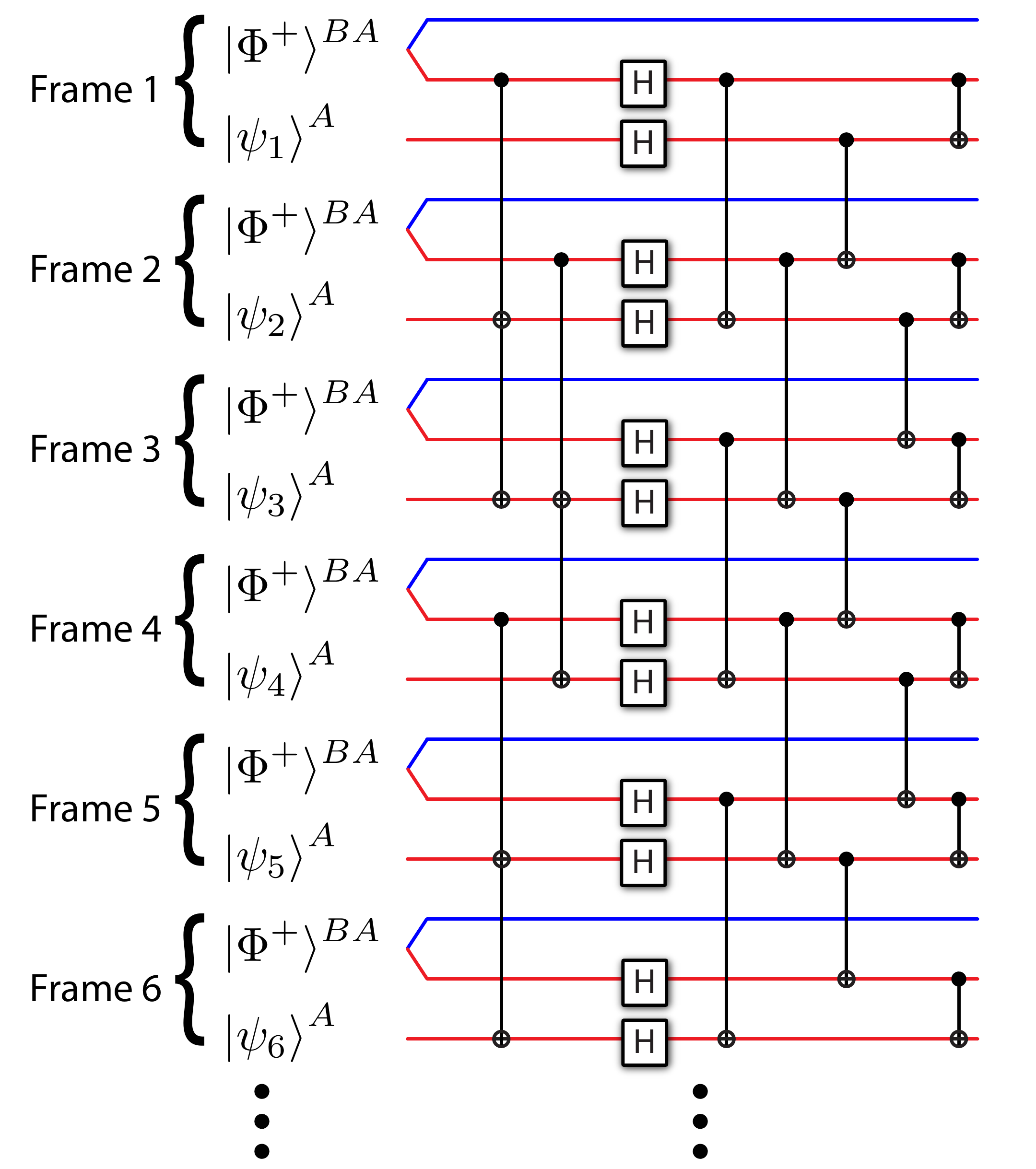}%
\caption{(Color online) The finite-depth encoding circuit for the
entanglement-assisted quantum convolutional code in
Example~\ref{ex:fefd-example}. The above operations in reverse order give a
valid decoding circuit.}%
\label{fig:example-eaqcc-fefd}%
\end{center}
\end{figure}
Let us begin with a polynomial matrix that stabilizes the unencoded state:%
\[
\left[  \left.
\begin{array}
[c]{ccc}%
1 & 1 & 0\\
0 & 0 & 0
\end{array}
\right\vert
\begin{array}
[c]{ccc}%
0 & 0 & 0\\
1 & 1 & 0
\end{array}
\right]  .
\]
Alice possesses the two qubits on the \textquotedblleft
right\textquotedblright\ and Bob possesses the qubit on the \textquotedblleft
left.\textquotedblright\ We label the middle qubit as \textquotedblleft qubit
one\textquotedblright\ and the rightmost qubit as \textquotedblleft qubit
two.\textquotedblright\ Alice performs a CNOT\ from qubit one to qubit two in
a delayed frame and a CNOT\ from qubit one to qubit two in a frame delayed by
two. The stabilizer becomes%
\[
\left[  \left.
\begin{array}
[c]{ccc}%
1 & 1 & 0\\
0 & 0 & 0
\end{array}
\right\vert
\begin{array}
[c]{ccc}%
0 & 0 & 0\\
1 & 1 & D+D^{2}%
\end{array}
\right]  .
\]
Alice performs Hadamard gates on both of her qubits. The stabilizer becomes%
\[
\left[  \left.
\begin{array}
[c]{ccc}%
1 & 0 & 0\\
0 & 1 & D+D^{2}%
\end{array}
\right\vert
\begin{array}
[c]{ccc}%
0 & 1 & 0\\
1 & 0 & 0
\end{array}
\right]  .
\]
Alice performs a CNOT\ from qubit one to qubit two in a delayed frame. The
stabilizer becomes%
\[
\left[  \left.
\begin{array}
[c]{ccc}%
1 & 0 & 0\\
0 & D & D+D^{2}%
\end{array}
\right\vert
\begin{array}
[c]{ccc}%
0 & 1 & D\\
1 & 0 & 0
\end{array}
\right]  .
\]
Alice performs a CNOT\ from qubit two to qubit one in a delayed frame. The
stabilizer becomes%
\[
\left[  \left.
\begin{array}
[c]{ccc}%
1 & 0 & 0\\
0 & D & 1+D+D^{2}%
\end{array}
\right\vert
\begin{array}
[c]{ccc}%
0 & 1+D^{2} & D\\
1 & 0 & 0
\end{array}
\right]  .
\]
Alice performs a CNOT\ from qubit one to qubit two. The stabilizer becomes%
\[
\left[  \left.
\begin{array}
[c]{ccc}%
1 & 0 & 0\\
0 & 1+D^{2} & 1+D+D^{2}%
\end{array}
\right\vert
\begin{array}
[c]{ccc}%
0 & 1+D^{2} & 1+D+D^{2}\\
1 & 0 & 0
\end{array}
\right]  .
\]
A row operation that switches the first row with the second row gives the
following stabilizer:%
\[
\left[  \left.
\begin{array}
[c]{ccc}%
0 & 1+D^{2} & 1+D+D^{2}\\
1 & 0 & 0
\end{array}
\right\vert
\begin{array}
[c]{ccc}%
1 & 0 & 0\\
0 & 1+D^{2} & 1+D+D^{2}%
\end{array}
\right]  .
\]
The entries on Alice's side of the above stabilizer have equivalent
error-correcting properties to the quantum check matrix in
(\ref{eq:desired-QCM-first-example}). Figure~\ref{fig:example-eaqcc-fefd}%
\ illustrates how the above operations encode a stream of ebits and
information qubits for our example.
\end{example}

\subsubsection{Discussion}

Codes in the first class are more useful in practice than those in the second
because their encoding and decoding circuits are finite depth. An uncorrected
error propagates only to a finite number of information qubits in the decoded
qubit stream. Codes in the first class therefore do not require any
assumptions about noiseless encoding or decoding.

The assumption about the invariant factors in the Smith form of $H_{1}\left(
D\right)  H_{2}^{T}\left(  D^{-1}\right)  $\ holds only for some classical
check matrices. Only a subclass of classical codes satisfy this assumption,
but it still expands the set of available quantum codes beyond those whose
check matrices $H_{1}\left(  D\right)  $ and $H_{2}\left(  D\right)  $ are
orthogonal. We need further techniques to handle the classical codes for which
this assumption does not hold. The following sections provide these further
techniques to handle a larger class of entanglement-assisted quantum
convolutional codes.

\section{Infinite-Depth Clifford Operations}

\label{sec:infinite-depth-ops}We now introduce a new type of operation, an
infinite-depth operation, to the set of operations in the shift-invariant
Clifford group available for encoding and decoding quantum convolutional
codes. We require infinite-depth operations to expand the set of classical
convolutional codes that we can import for quantum convolutional coding.

\begin{definition}
An \textit{infinite-depth operation} can transform a finite-weight stabilizer
generator to one with infinite weight (but does not necessarily do so to every
finite-weight generator).
\end{definition}

A decoding circuit with infinite-depth operations on qubits sent over the
noisy channel is undesirable because it spreads uncorrected errors infinitely
into the decoded information qubit stream. But an encoding circuit with
infinite-depth operations is acceptable if we assume a communication paradigm
in which the only noisy process is the noisy quantum channel.

We later show several examples of circuits that include infinite-depth
operations. Infinite-depth operations expand the possibilities for quantum
convolutional circuits in much the same way that incorporating feedback
expands the possibilities for classical convolutional circuits.

We illustrate the details of several infinite-depth operations for use in an
entanglement-assisted quantum convolutional code. We first provide some
specific examples of infinite-depth operations and then show how to realize an
arbitrary infinite-depth operation.

We consider both the stabilizer and the logical operators for the information
qubits in our analysis. Tracking both of these sets of generators is necessary
for determining the proper decoding circuit when including infinite-depth operations.

\subsection{Examples of Infinite-Depth Operations}

Our first example of an infinite-depth operation involves a stream of
information qubits and ancilla qubits. We divide the stream into frames of
three qubits where each frame has two ancilla qubits and one information
qubit. The following two generators and each of their three-qubit shifts
stabilize the qubit stream:%
\begin{equation}
\cdots\left\vert
\begin{array}
[c]{ccc}%
I & I & I\\
I & I & I
\end{array}
\right\vert
\begin{array}
[c]{ccc}%
Z & I & I\\
I & Z & I
\end{array}
\left\vert
\begin{array}
[c]{ccc}%
I & I & I\\
I & I & I
\end{array}
\right\vert \cdots\label{eq:id-unencoded-Paulis}%
\end{equation}
The binary polynomial matrix corresponding to this stabilizer is as follows:%
\begin{equation}
\left[  \left.
\begin{array}
[c]{ccc}%
1 & 0 & 0\\
0 & 1 & 0
\end{array}
\right\vert
\begin{array}
[c]{ccc}%
0 & 0 & 0\\
0 & 0 & 0
\end{array}
\right]  . \label{eq:id-unencoded-qubits}%
\end{equation}
We obtain any Pauli sequence in the stabilizer by multiplying the above rows
by a power of $D$ and applying the inverse of the P2B isomorphism. The logical
operators for the information qubits are as follows:%
\[
\cdots\left\vert
\begin{array}
[c]{ccc}%
I & I & I\\
I & I & I
\end{array}
\right\vert
\begin{array}
[c]{ccc}%
I & I & X\\
I & I & Z
\end{array}
\left\vert
\begin{array}
[c]{ccc}%
I & I & I\\
I & I & I
\end{array}
\right\vert \cdots
\]
They also admit a description with a binary polynomial matrix:%
\begin{equation}
\left[  \left.
\begin{array}
[c]{ccc}%
0 & 0 & 0\\
0 & 0 & 1
\end{array}
\right\vert
\begin{array}
[c]{ccc}%
0 & 0 & 1\\
0 & 0 & 0
\end{array}
\right]  . \label{eq-unencoded-info-qubits}%
\end{equation}
We refer to the above matrix as the \textquotedblleft information-qubit
matrix.\textquotedblright

\subsubsection{Encoding}

Suppose we would like to encode the above stream so that the following
generators stabilize it:%
\[
\cdots\left\vert
\begin{array}
[c]{ccc}%
I & I & I\\
I & I & I
\end{array}
\right\vert
\begin{array}
[c]{ccc}%
X & X & X\\
Z & Z & I
\end{array}
\left\vert
\begin{array}
[c]{ccc}%
X & X & I\\
I & I & I
\end{array}
\right\vert \cdots,
\]
or equivalently, the following binary polynomial matrix stabilizes it:%
\begin{equation}
\left[  \left.
\begin{array}
[c]{ccc}%
0 & 0 & 0\\
1 & 1 & 0
\end{array}
\right\vert
\begin{array}
[c]{ccc}%
D+1 & D+1 & 1\\
0 & 0 & 0
\end{array}
\right]  . \label{eq:desired-stabilizer}%
\end{equation}

We encode the above stabilizer using a combination of finite-depth operations
and an infinite-depth operation. We perform a Hadamard on the first qubit in
each frame and follow with a CNOT\ from the first qubit to the second and
third qubits in each frame. These operations transform the matrix in
(\ref{eq:id-unencoded-qubits}) to the following matrix%
\[
\left[  \left.
\begin{array}
[c]{ccc}%
0 & 0 & 0\\
1 & 1 & 0
\end{array}
\right\vert
\begin{array}
[c]{ccc}%
1 & 1 & 1\\
0 & 0 & 0
\end{array}
\right]  ,
\]
or equivalently transform the generators in (\ref{eq:id-unencoded-Paulis})\ to
the following generators:%
\[
\cdots\left\vert
\begin{array}
[c]{ccc}%
I & I & I\\
I & I & I
\end{array}
\right\vert
\begin{array}
[c]{ccc}%
X & X & X\\
Z & Z & I
\end{array}
\left\vert
\begin{array}
[c]{ccc}%
I & I & I\\
I & I & I
\end{array}
\right\vert \cdots
\]
The information-qubit matrix becomes%
\[
\left[  \left.
\begin{array}
[c]{ccc}%
0 & 0 & 0\\
1 & 0 & 1
\end{array}
\right\vert
\begin{array}
[c]{ccc}%
0 & 0 & 1\\
0 & 0 & 0
\end{array}
\right]  .
\]
We now perform an infinite-depth operation: a CNOT\ from the third qubit in
one frame to the third qubit in a delayed frame and repeat this operation for
all following frames. Figure~\ref{fig:inf-depth-simple}\ shows this operation
acting on our stream of qubits with three qubits per frame.%
\begin{figure}
[ptb]
\begin{center}
\includegraphics[
natheight=9.786200in,
natwidth=3.673700in,
height=3.1808in,
width=1.2047in
]%
{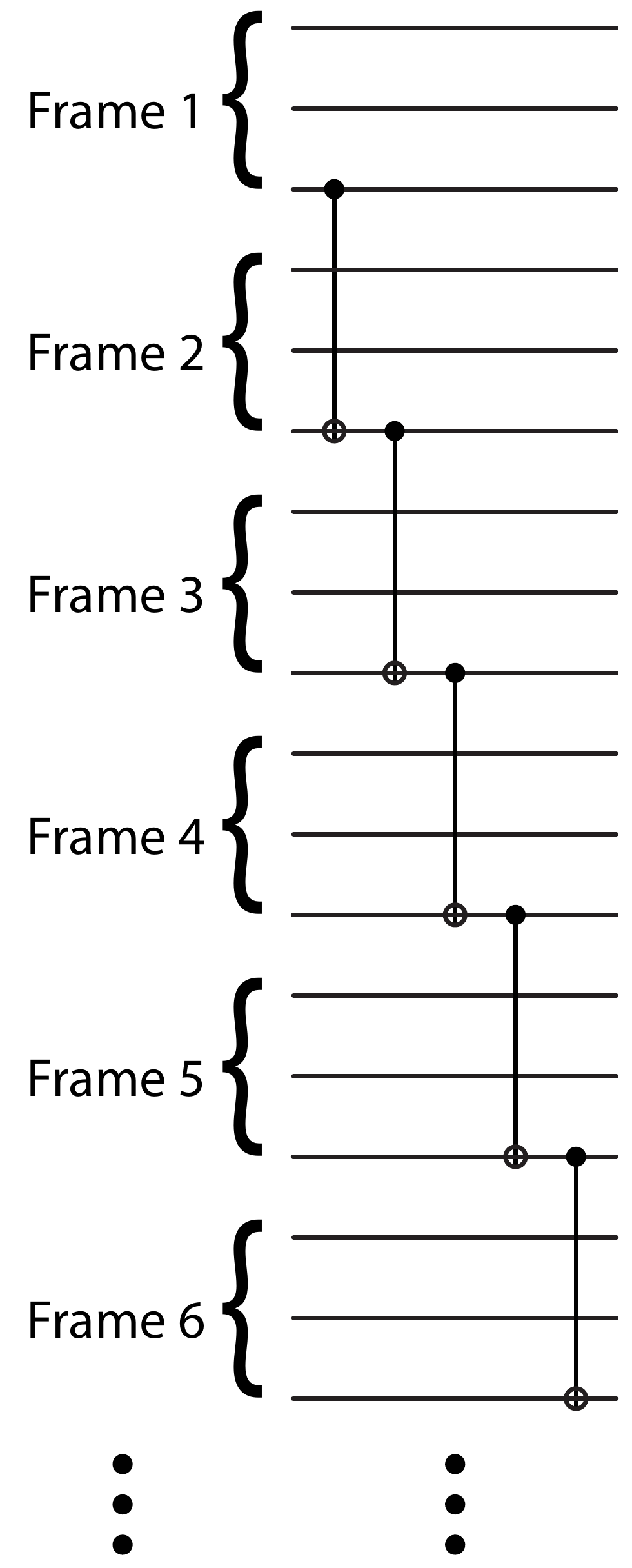}%
\caption{An example of an infinite-depth operation. A sequence of CNOT\ gates
acts on the third qubit of every frame. This infinite-depth operation
effectively multiplies the third column of the \textquotedblleft
X\textquotedblright\ side of the binary polynomial matrix by the rational
polynomial $1/\left(  1+D\right)  $ and multiplies the third column of the
\textquotedblleft Z\textquotedblright\ side of the binary polynomial matrix by
$1+D^{-1}$.}%
\label{fig:inf-depth-simple}%
\end{center}
\end{figure}
The effect of this operation is to translate the above stabilizer generators
as follows:%
\[
\cdots\left\vert
\begin{array}
[c]{ccc}%
I & I & I\\
I & I & I
\end{array}
\right\vert
\begin{array}
[c]{ccc}%
X & X & X\\
Z & Z & I
\end{array}
\left\vert
\begin{array}
[c]{ccc}%
I & I & X\\
I & I & I
\end{array}
\right\vert \left.
\begin{array}
[c]{ccc}%
I & I & X\\
I & I & I
\end{array}
\right\vert \cdots
\]
The first generator above and each of its three-qubits shifts is an
infinite-weight generator if the above sequence of CNOTs acts on the entire
countably-infinite qubit stream. We represent the above stabilizer with the
binary \textit{rational} polynomial matrix%
\begin{equation}
\left[  \left.
\begin{array}
[c]{ccc}%
0 & 0 & 0\\
1 & 1 & 0
\end{array}
\right\vert
\begin{array}
[c]{ccc}%
1 & 1 & 1/\left(  1+D\right) \\
0 & 0 & 0
\end{array}
\right]  , \label{eq:id-encoded-stabilizer}%
\end{equation}
where $1/\left(  1+D\right)  =1+D+D^{2}+\ldots$ is a repeating fraction. The
operation is infinite-depth because it translates the original finite-weight
stabilizer generator to one with infinite weight.

It is possible to perform a row operation that multiplies the first row by
$D+1$. This operation gives a stabilizer matrix that is equivalent to the
desired stabilizer in (\ref{eq:desired-stabilizer}). The receiver of the
encoded qubits measures the finite-weight stabilizer generators in
(\ref{eq:desired-stabilizer}) to diagnose errors. These measurements do not
disturb the information qubits because they also stabilize the encoded stream.

The above encoding operations transform the information-qubit matrix as
follows:%
\begin{equation}
\left[  \left.
\begin{array}
[c]{ccc}%
0 & 0 & 0\\
1 & 0 & 1+D^{-1}%
\end{array}
\right\vert
\begin{array}
[c]{ccc}%
0 & 0 & 1/\left(  1+D\right) \\
0 & 0 & 0
\end{array}
\right]  . \label{eq:id-encoded-info-qubits}%
\end{equation}
The infinite-depth operation on the third qubit has an effect on the
\textquotedblleft Z\textquotedblright\ or left side of the information-qubit
matrix as illustrated in the second row of the above matrix. The effect is to
multiply the third column of the \textquotedblleft Z\textquotedblright\ matrix
by $f\left(  D^{-1}\right)  $ if the operation multiplies the third column of
the \textquotedblleft X\textquotedblright\ matrix by $1/f\left(  D\right)  $.
This corresponding action on the \textquotedblleft Z\textquotedblright\ side
occurs because the commutation relations of the Pauli operators remain
invariant under quantum gates, or equivalently, the shifted symplectic product
remains invariant under column operations. The original shifted symplectic
product for the logical operators is one, and it remains as one because
$f((D^{-1})^{-1})/f\left(  D\right)  =1$.

\subsubsection{Decoding}

We perform finite-depth operations to decode the stream of information qubits.
Begin with the stabilizer and information-qubit matrix in
(\ref{eq:id-encoded-stabilizer})\ and (\ref{eq:id-encoded-info-qubits})
respectively. Perform a CNOT\ from the first qubit to the second qubit. The
stabilizer becomes%
\[
\left[  \left.
\begin{array}
[c]{ccc}%
0 & 0 & 0\\
0 & 1 & 0
\end{array}
\right\vert
\begin{array}
[c]{ccc}%
1 & 0 & 1/\left(  1+D\right) \\
0 & 0 & 0
\end{array}
\right]  ,
\]
and the information-qubit matrix does not change. Perform a CNOT\ from the
third qubit to the first qubit in the same frame and in a delayed frame. These
gates multiply column three in the \textquotedblleft X\textquotedblright%
\ matrix by $1+D$ and add the result to column one. The gates also multiply
column one in the \textquotedblleft Z\textquotedblright\ matrix by $1+D^{-1}$
and add the result to column three. The effect is as follows on both the
stabilizer%
\begin{equation}
\left[  \left.
\begin{array}
[c]{ccc}%
0 & 0 & 0\\
0 & 1 & 0
\end{array}
\right\vert
\begin{array}
[c]{ccc}%
0 & 0 & 1/\left(  1+D\right) \\
0 & 0 & 0
\end{array}
\right]  , \label{eq:id-decoded-stabilizer}%
\end{equation}
and the information-qubit matrix%
\begin{equation}
\left[  \left.
\begin{array}
[c]{ccc}%
0 & 0 & 0\\
1 & 0 & 0
\end{array}
\right\vert
\begin{array}
[c]{ccc}%
1 & 0 & 1/\left(  1+D\right) \\
0 & 0 & 0
\end{array}
\right]  . \label{eq:id-decoded-info}%
\end{equation}
We can multiply the logical operators by any element of the stabilizer and
obtain an equivalent logical operator \cite{thesis97gottesman}. We perform
this multiplication in the \textquotedblleft binary-polynomial
picture\textquotedblright\ by adding the first row of the stabilizer in
(\ref{eq:id-decoded-stabilizer}) to the first row of (\ref{eq:id-decoded-info}%
). The information-qubit matrix becomes%
\begin{equation}
\left[  \left.
\begin{array}
[c]{ccc}%
0 & 0 & 0\\
1 & 0 & 0
\end{array}
\right\vert
\begin{array}
[c]{ccc}%
1 & 0 & 0\\
0 & 0 & 0
\end{array}
\right]  , \label{eq:encoded-info-qubit-1st-ex}%
\end{equation}
so that the resulting logical operators act only on the first qubit of every
frame. We have successfully decoded the information qubits with finite-depth
operations. The information qubits teleport coherently
\cite{prl2004harrow,wilde:060303}\ from being the third qubit of each frame as
in (\ref{eq-unencoded-info-qubits}) to being the first qubit of each frame as
in (\ref{eq:encoded-info-qubit-1st-ex}). We exploit the above method of
encoding with infinite-depth operations and decoding with finite-depth
operations for the class of entanglement-assisted quantum convolutional codes
in Section~\ref{sec:eaqcc-iefd}.

\subsection{General Infinite-Depth Operations}

We discuss the action of a general infinite-depth operation on two weight-one
\textquotedblleft X\textquotedblright\ and \textquotedblleft
Z\textquotedblright\ Pauli sequences where each frame has one Pauli matrix.
Our analysis then determines the effect of an infinite-depth operation on an
arbitrary stabilizer or information-qubit matrix. The generators in the
\textquotedblleft Pauli picture\textquotedblright\ are as follows:%
\begin{equation}
\cdots\left\vert
\begin{array}
[c]{c}%
I\\
I
\end{array}
\right.  \left\vert
\begin{array}
[c]{c}%
X\\
Z
\end{array}
\right\vert \left.
\begin{array}
[c]{c}%
I\\
I
\end{array}
\right\vert \cdots, \label{eq:ie-unencoded-Paulis}%
\end{equation}
or as follows in the \textquotedblleft binary-polynomial
picture\textquotedblright:%
\[
\left[  \left.
\begin{array}
[c]{c}%
0\\
1
\end{array}
\right\vert
\begin{array}
[c]{c}%
1\\
0
\end{array}
\right]  .
\]
An infinite-depth $1/f\left(  D\right)  $ operation, where $f\left(  D\right)
$ is an arbitrary polynomial, should transform the above matrix to the
following one:%
\[
\left[  \left.
\begin{array}
[c]{c}%
0\\
f\left(  D^{-1}\right)
\end{array}
\right\vert
\begin{array}
[c]{c}%
1/f\left(  D\right) \\
0
\end{array}
\right]  .
\]
A circuit that performs this transformation preserves the shifted symplectic
product because $f\left(  D^{-1}\right)  \cdot1/f\left(  D^{-1}\right)  =1$.
The circuit should operate on a few qubits at a time and should be
shift-invariant---the same device or physical routines implement it.

First perform the long division expansion of binary rational polynomial
$1/f\left(  D\right)  $. This expansion has a particular repeating pattern
with period $l$. For example, suppose that $f\left(  D\right)  =1+D+D^{3}$.
Its long-division expansion is $1+D+D^{2}+D^{4}+D^{7}+D^{8}+D^{9}%
+D^{11}+\ldots$ and exhibits a repeating pattern with period seven. We want a
circuit that realizes the following Pauli generators%
\begin{equation}
\cdots\left\vert
\begin{array}
[c]{c}%
I\\
Z
\end{array}
\right.  \left\vert
\begin{array}
[c]{c}%
I\\
I
\end{array}
\right\vert
\begin{array}
[c]{c}%
I\\
Z
\end{array}
\left\vert
\begin{array}
[c]{c}%
X\\
Z
\end{array}
\right\vert
\begin{array}
[c]{c}%
X\\
I
\end{array}
\left\vert
\begin{array}
[c]{c}%
X\\
I
\end{array}
\right\vert
\begin{array}
[c]{c}%
I\\
I
\end{array}
\left\vert
\begin{array}
[c]{c}%
X\\
I
\end{array}
\right\vert \left.
\begin{array}
[c]{c}%
I\\
I
\end{array}
\right\vert \cdots, \label{eq:ie-desired-paulis}%
\end{equation}
where the pattern in the $X$ matrices is the same as the repeating polynomial
$1/f\left(  D\right)  $ and continues infinitely to the right, and the pattern
on the $Z$ matrices is the same as that in $f\left(  D^{-1}\right)  $\ and
terminates at the left. The above Pauli sequence is equivalent to the
following binary rational polynomial matrix:%
\[
\left[  \left.
\begin{array}
[c]{c}%
0\\
1+D^{-1}+D^{-3}%
\end{array}
\right\vert
\begin{array}
[c]{c}%
1/\left(  1+D+D^{3}\right) \\
0
\end{array}
\right]  .
\]

We now discuss a method that realizes an arbitrary rational polynomial
$1/f\left(  D\right)  $ as an infinite-depth operation. Our method for
encoding the generators in (\ref{eq:ie-desired-paulis}) from those in
(\ref{eq:ie-unencoded-Paulis}) consists of a \textquotedblleft
sliding-window\textquotedblright\ technique that determines transformation
rules for the circuit. The circuit is an additive, shift-invariant filtering
operation. It resembles an infinite-impulse response filter because the
sequence it produces extends infinitely. In general, the number $N$\ of qubits
that the encoding unitary operates on is as follows%
\[
N=\deg\left(  f\left(  D\right)  \right)  -\mathrm{del}\left(  f\left(
D\right)  \right)  +1,
\]
where $\deg\left(  f\left(  D\right)  \right)  $ and $\mathrm{del}\left(
f\left(  D\right)  \right)  $ are the respective highest and lowest powers of
polynomial $f\left(  D\right)  $. Therefore, our exemplary encoding unitary
operates on four qubits at a time. We delay the original sequence
in\ (\ref{eq:ie-unencoded-Paulis}) by three frames. These initial frames are
\textquotedblleft scratch\textquotedblright\ frames that give the encoding
unitary enough \textquotedblleft room\textquotedblright\ to generate the
desired Paulis in (\ref{eq:ie-desired-paulis}). The first set of
transformation rules is as follows%
\begin{equation}
\left.
\begin{array}
[c]{c}%
I\\
I
\end{array}
\right\vert
\begin{array}
[c]{c}%
I\\
I
\end{array}
\left\vert
\begin{array}
[c]{c}%
I\\
I
\end{array}
\right\vert
\begin{array}
[c]{c}%
X\\
Z
\end{array}
\rightarrow\left.
\begin{array}
[c]{c}%
I\\
Z
\end{array}
\right\vert
\begin{array}
[c]{c}%
I\\
I
\end{array}
\left\vert
\begin{array}
[c]{c}%
I\\
Z
\end{array}
\right\vert
\begin{array}
[c]{c}%
X\\
Z
\end{array}
, \label{eq:first-rule}%
\end{equation}
and generates the first four elements of the pattern in
(\ref{eq:ie-desired-paulis}). Now that the encoding unitary has acted on the
first four frames, we need to shift our eyes to the right by one frame in the
sequence in (\ref{eq:ie-desired-paulis})\ to determine the next set of rules.
So we shift the above outputs by one frame to the \textit{left} (assuming that
only identity matrices lie to the right) and determine the next set of
transformation rules that generate the next elements of the sequence in
(\ref{eq:ie-desired-paulis}):%
\[
\left.
\begin{array}
[c]{c}%
I\\
I
\end{array}
\right\vert
\begin{array}
[c]{c}%
I\\
Z
\end{array}
\left\vert
\begin{array}
[c]{c}%
X\\
Z
\end{array}
\right\vert
\begin{array}
[c]{c}%
I\\
I
\end{array}
\rightarrow\left.
\begin{array}
[c]{c}%
I\\
I
\end{array}
\right\vert
\begin{array}
[c]{c}%
I\\
Z
\end{array}
\left\vert
\begin{array}
[c]{c}%
X\\
Z
\end{array}
\right\vert
\begin{array}
[c]{c}%
X\\
I
\end{array}
.
\]
Shift the above outputs to the left by one frame to determine the next set of
transformation rules:%
\[
\left.
\begin{array}
[c]{c}%
I\\
Z
\end{array}
\right\vert
\begin{array}
[c]{c}%
X\\
Z
\end{array}
\left\vert
\begin{array}
[c]{c}%
X\\
I
\end{array}
\right\vert
\begin{array}
[c]{c}%
I\\
I
\end{array}
\rightarrow\left.
\begin{array}
[c]{c}%
I\\
Z
\end{array}
\right\vert
\begin{array}
[c]{c}%
X\\
Z
\end{array}
\left\vert
\begin{array}
[c]{c}%
X\\
I
\end{array}
\right\vert
\begin{array}
[c]{c}%
X\\
I
\end{array}
.
\]
We obtain the rest of the transformation rules by continuing this sliding
process, and we stop when the pattern in the sequence in
(\ref{eq:ie-desired-paulis}) begins to repeat:%
\[
\left.
\begin{array}
[c]{c}%
X\\
Z\\
X\\
X\\
I\\
X
\end{array}
\right\vert
\begin{array}
[c]{c}%
X\\
I\\
X\\
I\\
X\\
I
\end{array}
\left\vert
\begin{array}
[c]{c}%
X\\
I\\
I\\
X\\
I\\
I
\end{array}
\right\vert
\begin{array}
[c]{c}%
I\\
I\\
I\\
I\\
I\\
I
\end{array}
\rightarrow\left.
\begin{array}
[c]{c}%
X\\
Z\\
X\\
X\\
I\\
X
\end{array}
\right\vert
\begin{array}
[c]{c}%
X\\
I\\
X\\
I\\
X\\
I
\end{array}
\left\vert
\begin{array}
[c]{c}%
X\\
I\\
I\\
X\\
I\\
I
\end{array}
\right\vert
\begin{array}
[c]{c}%
I\\
I\\
X\\
I\\
I\\
X
\end{array}
.
\]
The above set of rules determines the encoding unitary and only a few of them
are actually necessary. We can multiply the rules together to form equivalent
rules because the circuit obeys additivity (in the \textquotedblleft
binary-polynomial picture\textquotedblright). The rules become as follows
after rearranging into a standard form:%
\[
\left.
\begin{array}
[c]{c}%
Z\\
I\\
I\\
I\\
X\\
I\\
I\\
I
\end{array}
\right\vert
\begin{array}
[c]{c}%
I\\
Z\\
I\\
I\\
I\\
X\\
I\\
I
\end{array}
\left\vert
\begin{array}
[c]{c}%
I\\
I\\
Z\\
I\\
I\\
I\\
X\\
I
\end{array}
\right\vert
\begin{array}
[c]{c}%
I\\
I\\
I\\
Z\\
I\\
I\\
I\\
X
\end{array}
\rightarrow\left.
\begin{array}
[c]{c}%
Z\\
I\\
I\\
Z\\
X\\
I\\
I\\
I
\end{array}
\right\vert
\begin{array}
[c]{c}%
I\\
Z\\
I\\
I\\
I\\
X\\
I\\
I
\end{array}
\left\vert
\begin{array}
[c]{c}%
I\\
I\\
Z\\
Z\\
I\\
I\\
X\\
I
\end{array}
\right\vert
\begin{array}
[c]{c}%
I\\
I\\
I\\
Z\\
X\\
I\\
X\\
X
\end{array}
.
\]
A CNOT\ from qubit one to qubit four and a CNOT\ from qubit three to qubit
four suffice to implement this circuit. We repeatedly apply these operations
shifting by one frame at a time to implement the infinite-depth operation. We
could have observed that these gates suffice to implement the
\textquotedblleft Z\textquotedblright\ transformation in the first set of
transformation rules in (\ref{eq:first-rule}), but we wanted to show how this
method generates the full periodic \textquotedblleft X\textquotedblright%
\ sequence in (\ref{eq:ie-desired-paulis}). Figure~\ref{fig:inf-depth-example}%
\ shows how the above encoding unitary acts on a stream of quantum
information.%
\begin{figure}
[ptb]
\begin{center}
\includegraphics[
natheight=4.633700in,
natwidth=5.233800in,
height=2.7484in,
width=3.1012in
]%
{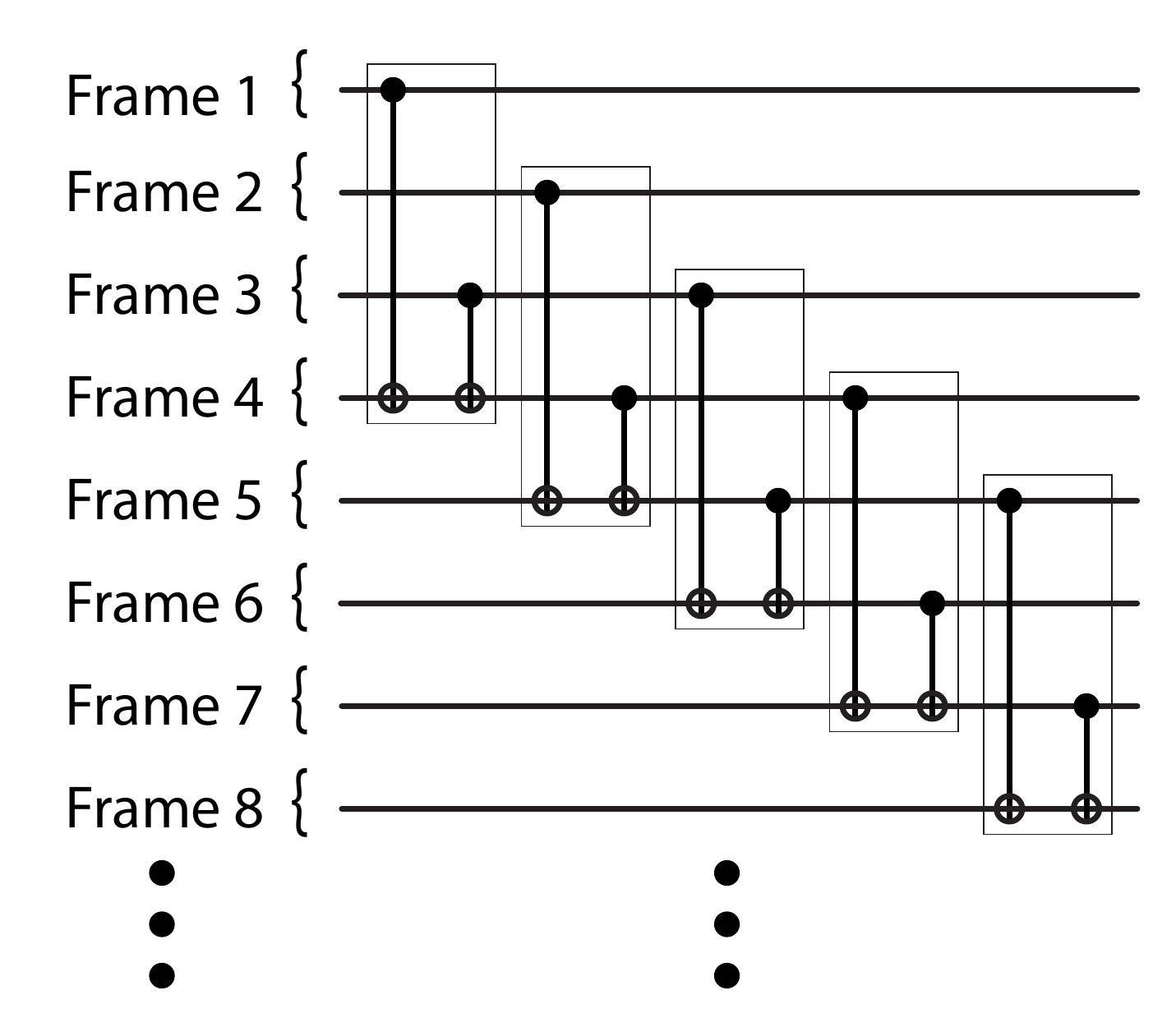}%
\caption{Another example of an infinite-depth operation. An infinite-depth
operation acts on qubit $i$ in every frame. This particular infinite-depth
operation multiplies column $i$ on the \textquotedblleft X\textquotedblright%
\ side of the binary polynomial matrix by $1/\left(  1+D+D^{3}\right)  $ and
multiplies column $i$ on the \textquotedblleft Z\textquotedblright\ side of
the binary polynomial matrix by $1+D^{-1}+D^{-3}$.}%
\label{fig:inf-depth-example}%
\end{center}
\end{figure}

We can determine the encoding unitary for an arbitrary rational polynomial
$1/f\left(  D\right)  $ using a similar method. Suppose that $\mathrm{del}%
\left(  f\left(  D\right)  \right)  =n$ and suppose $n\neq0$ as in the above
case. First delay or advance the frames if $n>0$ or if $n<0$ respectively.
Determine the CNOT\ gates that transform the \textquotedblleft
Z\textquotedblright\ Pauli sequence%
\[
\left[  \left.
\begin{array}
[c]{c}%
1
\end{array}
\right\vert
\begin{array}
[c]{c}%
0
\end{array}
\right]
\]
to%
\[
\left[  \left.
\begin{array}
[c]{c}%
D^{n}f\left(  D^{-1}\right)
\end{array}
\right\vert
\begin{array}
[c]{c}%
0
\end{array}
\right]  .
\]
These CNOT\ gates form the encoding circuit that transform both the
\textquotedblleft X\textquotedblright\ and \textquotedblleft
Z\textquotedblright\ Pauli sequences. We perform the encoding unitary, shift
by one frame, perform it again, and keep repeating. Our method encodes any
arbitrary polynomial $1/f\left(  D\right)  $ on the \textquotedblleft
X\textquotedblright\ side and $f\left(  D^{-1}\right)  $ on the
\textquotedblleft Z\textquotedblright\ side.

We can implement the \textquotedblleft time-reversed\textquotedblright%
\ polynomial $1/f\left(  D^{-1}\right)  $ on the \textquotedblleft
X\textquotedblright\ side by first delaying the frames by $m=\deg\left(
f\left(  D\right)  \right)  -\mathrm{del}\left(  f\left(  D\right)  \right)  $
frames and performing the circuit corresponding to $1/D^{m}\left(  f\left(
D^{-1}\right)  \right)  $. These operations implement the circuit $D^{m}%
/D^{m}\left(  f\left(  D^{-1}\right)  \right)  =1/f\left(  D^{-1}\right)  $.

\subsection{Infinite-Depth Operations in Practice}

We assume above that each of the infinite-depth operations acts on the entire
countably-infinite stream of qubits. In practice, each infinite-depth
operation acts on a finite number of qubits at a time so that the encoding and
decoding circuits operate in an \textquotedblleft online\textquotedblright%
\ manner. Therefore, each infinite-depth operation approximates its
corresponding rational polynomial. This approximation does not pose a barrier
to implementation. We can implement each of the above infinite-depth
operations by padding the initial qubits of the information qubit stream with
some \textquotedblleft scratch\textquotedblright\ qubits. We first transmit
these \textquotedblleft scratch\textquotedblright\ qubits that contain no
useful quantum information so that the later information qubits enjoy the full
protection of the code. These scratch qubits do not affect the asymptotic rate
of the code and merely serve as a convenience for implementing the
infinite-depth operations. From now on, we adhere to describing infinite-depth
operations with binary rational polynomials because it is more convenient to
do so mathematically.

\subsection{Entanglement-Assisted Quantum Convolutional Codes with
Infinite-Depth Operations}

In the section that follows, our entanglement-assisted quantum convolutional
codes have infinite-depth operations in their encoding circuits. This
possibility is acceptable because the entanglement-assisted communication
paradigm assumes that noiseless encoding is possible and that the receiver's
half of the ebits are noiseless. We later briefly discuss the effects of
relaxing this assumption in a realistic system.

Our decoding circuits in the second class of codes perform finite-depth
operations. Some of our decoding circuits are not the exact inverse of their
corresponding encoding circuits, but the decoding circuits invert the effect
of the encoding circuits because they produce the original stream of
information qubits at their output.

\section{Entanglement-Assisted Quantum Convolutional Codes with Infinite-Depth
Encoding and Finite-Depth Decoding Circuits}

\label{sec:eaqcc-iefd}This section details codes whose encoding circuits have
both infinite-depth and finite-depth operations. We therefore assume that
encoding is noiseless to eliminate the possibility of encoding errors
spreading infinitely into the encoded qubit stream. Their decoding circuits
require finite-depth operations only.

Just as with the previous class, this class of codes is determined by the
properties of their corresponding classical check matrices, as described in
the following lemma.

\begin{lemma}
\label{lemma:ieid}Suppose the Smith form of $E\left(  D\right)  $ does not
admit the form from Lemma~\ref{lemma:fefd}. Then the entanglement-assisted
quantum convolutional code has an encoding circuit with both infinite-depth
and finite-depth operations. Its decoding circuit has finite-depth operations.
\end{lemma}

%

\begin{proof}%
We perform all of the operations from Lemma~\ref{lemma:general-ops}. The Smith
form of $E\left(  D\right)  $ is in general as follows%
\[
A_{E}\left(  D\right)
\begin{bmatrix}
\Gamma_{1}\left(  D\right)  & 0 & 0\\
0 & \Gamma_{2}\left(  D\right)  & 0\\
0 & 0 & 0
\end{bmatrix}
B_{E}\left(  D\right)  ,
\]
where $A_{E}\left(  D\right)  $ is $\left(  n-k_{1}\right)  \times\left(
n-k_{1}\right)  $, $\Gamma_{1}\left(  D\right)  $ is an $s\times s$ diagonal
matrix whose entries are powers of $D$, $\Gamma_{2}\left(  D\right)  $ is a
$\left(  c-s\right)  \times\left(  c-s\right)  $ diagonal matrix whose entries
are arbitrary polynomials, and $B_{E}\left(  D\right)  $ is $\left(
n-k_{2}\right)  \times\left(  n-k_{2}\right)  $. Perform the row operations in
$A_{E}^{-1}\left(  D\right)  $ and the column operations in $B_{E}^{-1}\left(
D\right)  $ on the quantum check matrix in (\ref{eq:last-eq-fefd}). Counteract
the effect of the column operations on the identity matrix in the
\textquotedblleft X\textquotedblright\ matrix by performing row operations.
The quantum check matrix in (\ref{eq:last-eq-fefd}) becomes%
\[
\left[  \left.
\begin{array}
[c]{cccc}%
\Gamma_{1}\left(  D\right)  & 0 & 0 & F_{1}\left(  D\right) \\
0 & \Gamma_{2}\left(  D\right)  & 0 & F_{2}\left(  D\right) \\
0 & 0 & 0 & F_{3}\left(  D\right) \\
0 & 0 & 0 & 0
\end{array}
\right\vert
\begin{array}
[c]{cc}%
0 & 0\\
0 & 0\\
0 & 0\\
I & 0
\end{array}
\right]  ,
\]
where $F_{1}\left(  D\right)  $, $F_{2}\left(  D\right)  $, and $F_{3}\left(
D\right)  $ are the respective $s$, $c-s$, and $n-k_{1}-c$ rows of $A_{E}%
^{-1}\left(  D\right)  F\left(  D\right)  $. The Smith form of $F_{3}\left(
D\right)  $ is as follows%
\[
F_{3}\left(  D\right)  =A_{F_{3}}\left(  D\right)
\begin{bmatrix}
\Gamma_{F_{3}}\left(  D\right)  & 0
\end{bmatrix}
B_{F_{3}}\left(  D\right)  ,
\]
where $A_{F_{3}}\left(  D\right)  $ is $\left(  n-k_{1}-c\right)
\times\left(  n-k_{1}-c\right)  $, $\Gamma_{F_{3}}\left(  D\right)  $ is an
$\left(  n-k_{1}-c\right)  \times\left(  n-k_{1}-c\right)  $ diagonal matrix
whose entries are powers of $D$, and $B_{F_{3}}\left(  D\right)  $ is
$k_{2}\times k_{2}$. The entries of $\Gamma_{F_{3}}\left(  D\right)  $ are
powers of $D$ because the original check matrix $H_{2}\left(  D\right)  $ is
noncatastrophic and column and row operations with Laurent polynomials change
the invariant factors only by a power of $D$. Perform the row operations in
$A_{F_{3}}^{-1}\left(  D\right)  $ and the column operations in $B_{F_{3}%
}^{-1}\left(  D\right)  $. The quantum check matrix becomes%
\[
\left[  \left.
\begin{array}
[c]{ccccc}%
\Gamma_{1}\left(  D\right)  & 0 & 0 & F_{1a}^{^{\prime}}\left(  D\right)  &
F_{1b}^{^{\prime}}\left(  D\right) \\
0 & \Gamma_{2}\left(  D\right)  & 0 & F_{2a}^{^{\prime}}\left(  D\right)  &
F_{2b}^{^{\prime}}\left(  D\right) \\
0 & 0 & 0 & \Gamma_{F_{3}}\left(  D\right)  & 0\\
0 & 0 & 0 & 0 & 0
\end{array}
\right\vert
\begin{array}
[c]{cc}%
0 & 0\\
0 & 0\\
0 & 0\\
I & 0
\end{array}
\right]  ,
\]
where $F_{1a}^{^{\prime}}\left(  D\right)  $, $F_{1b}^{^{\prime}}\left(
D\right)  $, $F_{2a}^{^{\prime}}\left(  D\right)  $, $F_{2b}^{^{\prime}%
}\left(  D\right)  $ are the matrices resulting from the column operations in
$B_{F_{3}}^{-1}\left(  D\right)  $. Perform row operations from the entries in
$\Gamma_{F_{3}}\left(  D\right)  $ to the rows above it to clear the entries
in $F_{1a}^{^{\prime}}\left(  D\right)  $ and $F_{2a}^{^{\prime}}\left(
D\right)  $. Use Hadamard and CNOT\ gates to clear the entries in
$F_{1b}^{^{\prime}}\left(  D\right)  $. The quantum check matrix becomes%
\[
\left[  \left.
\begin{array}
[c]{ccccc}%
\Gamma_{1}\left(  D\right)  & 0 & 0 & 0 & 0\\
0 & \Gamma_{2}\left(  D\right)  & 0 & 0 & F_{2b}^{^{\prime}}\left(  D\right)
\\
0 & 0 & 0 & \Gamma_{F_{3}}\left(  D\right)  & 0\\
0 & 0 & 0 & 0 & 0
\end{array}
\right\vert
\begin{array}
[c]{cc}%
0 & 0\\
0 & 0\\
0 & 0\\
I & 0
\end{array}
\right]  .
\]
We can reduce $F_{2b}^{^{\prime}}\left(  D\right)  $ to a lower triangular
form with an algorithm consisting of column operations only. The algorithm
operates on the last $k_{2}+k_{1}-n+c$ columns. It is similar to the Smith
algorithm but does not involve row operations. Consider the first row of
$F_{2b}^{^{\prime}}\left(  D\right)  $. Perform column operations between the
different elements of the row to reduce it to one non-zero entry. Swap this
non-zero entry to the leftmost position. Perform the same algorithm on
elements $2,\ldots,k_{2}+k_{1}-n+c$ of the second row. Continue on for all
rows of $F_{2b}^{^{\prime}}\left(  D\right)  $ to reduce it to a matrix of the
following form\bigskip%
\[
F_{2b}^{^{\prime}}\left(  D\right)  \rightarrow%
\begin{bmatrix}%
\raisebox{0ex}[1.5ex]{$\overbrace{L(D)}^{c-s}$}%
&
\raisebox{0ex}[1.5ex]{$\overbrace{0}^{k_1+k_2-n+s}$}%
\end{bmatrix}
,
\]
where $L\left(  D\right)  $ is a lower triangular matrix. The above quantum
check matrix becomes%
\[
\left[  \left.
\begin{array}
[c]{cccccc}%
\Gamma_{1}\left(  D\right)  & 0 & 0 & 0 & 0 & 0\\
0 & \Gamma_{2}\left(  D\right)  & 0 & 0 & L\left(  D\right)  & 0\\
0 & 0 & 0 & \Gamma_{F_{3}}\left(  D\right)  & 0 & 0\\
0 & 0 & 0 & 0 & 0 & 0
\end{array}
\right\vert
\begin{array}
[c]{cc}%
0 & 0\\
0 & 0\\
0 & 0\\
I & 0
\end{array}
\right]  .
\]
We have completed decomposition of the first set of $s$ rows with $\Gamma
_{1}\left(  D\right)  $, the third set of $n-k_{1}-c$ rows with $\Gamma
_{F_{3}}\left(  D\right)  $, and rows $n-k_{1}+1,\ldots,n-k_{1}+s$ with the
identity matrix on the \textquotedblleft X\textquotedblright\ side.

We now consider an algorithm with infinite-depth operations to encode the
following submatrix of the above quantum check matrix:%
\begin{equation}
\left[  \left.
\begin{array}
[c]{cc}%
\Gamma_{2}\left(  D\right)  & L\left(  D\right) \\
0 & 0
\end{array}
\right\vert
\begin{array}
[c]{cc}%
0 & 0\\
I & 0
\end{array}
\right]  . \label{eq:ieid-desired-matrix}%
\end{equation}
We begin with a set of $c-s$ ebits and $c-s$ information qubits. The following
matrix stabilizes the ebits%
\[
\left[  \left.
\begin{array}
[c]{ccc}%
I & I & 0\\
0 & 0 & 0
\end{array}
\right\vert
\begin{array}
[c]{ccc}%
0 & 0 & 0\\
I & I & 0
\end{array}
\right]  ,
\]
and the following matrix represents the information qubits%
\[
\left[  \left.
\begin{array}
[c]{ccc}%
0 & 0 & I\\
0 & 0 & 0
\end{array}
\right\vert
\begin{array}
[c]{ccc}%
0 & 0 & 0\\
0 & 0 & I
\end{array}
\right]  ,
\]
where all matrices have dimension $\left(  c-s\right)  \times\left(
c-s\right)  $ and Bob possesses the $c-s$ qubits on the \textquotedblleft
left\textquotedblright\ and Alice possesses the $2\left(  c-s\right)  $ qubits
on the \textquotedblleft right.\textquotedblright\ We track both the
stabilizer and the information qubits as they progress through some encoding
operations. Alice performs CNOT\ and Hadamard gates on her $2\left(
c-s\right)  $ qubits. These gates multiply the middle $c-s$ columns of the
\textquotedblleft Z\textquotedblright\ matrix by $L\left(  D\right)  $ and add
the result to the last $c-s$ columns and multiply the last $c-s$ columns of
the \textquotedblleft X\textquotedblright\ matrix by $L^{T}\left(
D^{-1}\right)  $ and add the result to the middle $c-s$ columns. The
stabilizer becomes%
\[
\left[  \left.
\begin{array}
[c]{ccc}%
I & I & L\left(  D\right) \\
0 & 0 & 0
\end{array}
\right\vert
\begin{array}
[c]{ccc}%
0 & 0 & 0\\
I & I & 0
\end{array}
\right]  ,
\]
and the information-qubit matrix becomes%
\[
\left[  \left.
\begin{array}
[c]{ccc}%
0 & 0 & I\\
0 & 0 & 0
\end{array}
\right\vert
\begin{array}
[c]{ccc}%
0 & 0 & 0\\
0 & L^{T}\left(  D^{-1}\right)  & I
\end{array}
\right]  .
\]
Alice performs infinite-depth operations on her first $c-s$ qubits
corresponding to the rational polynomials $\gamma_{2,1}^{-1}\left(
D^{-1}\right)  $, $\ldots$, $\gamma_{2,c-s}^{-1}\left(  D^{-1}\right)  $ in
$\Gamma_{2}^{-1}\left(  D^{-1}\right)  $. The stabilizer matrix becomes%
\[
\left[  \left.
\begin{array}
[c]{ccc}%
I & \Gamma_{2}\left(  D\right)  & L\left(  D\right) \\
0 & 0 & 0
\end{array}
\right\vert
\begin{array}
[c]{ccc}%
0 & 0 & 0\\
I & \Gamma_{2}^{-1}\left(  D^{-1}\right)  & 0
\end{array}
\right]  ,
\]
and the information-qubit matrix becomes%
\[
\left[  \left.
\begin{array}
[c]{ccc}%
0 & 0 & I\\
0 & 0 & 0
\end{array}
\right\vert
\begin{array}
[c]{ccc}%
0 & 0 & 0\\
0 & L^{T}\left(  D^{-1}\right)  \Gamma_{2}^{-1}\left(  D^{-1}\right)  & I
\end{array}
\right]  .
\]
Alice's part of the above stabilizer matrix is equivalent to the quantum check
matrix in (\ref{eq:ieid-desired-matrix}) by row operations (premultiplying the
second set of rows in the stabilizer by $\Gamma_{2}\left(  D\right)  $.) Bob
can therefore make stabilizer measurements that have finite weight and that
are equivalent to the desired stabilizer.

We now describe a method to decode the above encoded stabilizer and
information-qubit matrix so that the information qubits appear at the output
of the decoding circuit. Bob performs Hadamard gates on his first and third
sets of $c-s$ qubits, performs CNOT\ gates from the first set of qubits to the
third set of qubits corresponding to the entries in $L\left(  D\right)  $, and
performs the Hadamard gates again. The stabilizer becomes%
\begin{equation}
\left[  \left.
\begin{array}
[c]{ccc}%
I & \Gamma_{2}\left(  D\right)  & 0\\
0 & 0 & 0
\end{array}
\right\vert
\begin{array}
[c]{ccc}%
0 & 0 & 0\\
I & \Gamma_{2}^{-1}\left(  D^{-1}\right)  & 0
\end{array}
\right]  , \label{eq:decoding-iefd}%
\end{equation}
and the information-qubit matrix becomes%
\[
\left[  \left.
\begin{array}
[c]{ccc}%
0 & 0 & I\\
0 & 0 & 0
\end{array}
\right\vert
\begin{array}
[c]{ccc}%
0 & 0 & 0\\
L^{T}\left(  D^{-1}\right)  & L^{T}\left(  D^{-1}\right)  \Gamma_{2}%
^{-1}\left(  D^{-1}\right)  & I
\end{array}
\right]  .
\]
Bob finishes decoding at this point because we can equivalently express the
information-qubit matrix as follows%
\[
\left[  \left.
\begin{array}
[c]{ccc}%
0 & 0 & I\\
0 & 0 & 0
\end{array}
\right\vert
\begin{array}
[c]{ccc}%
0 & 0 & 0\\
0 & 0 & I
\end{array}
\right]  ,
\]
by multiplying the last $c-s$ rows of the stabilizer by $L^{T}\left(
D^{-1}\right)  $ and adding to the last $c-s$ rows of the information-qubit matrix.

The overall procedure for encoding is to begin with a set of $c$ ebits,
$2\left(  n-c\right)  -k_{1}-k_{2}$ ancilla qubits, and $k_{1}+k_{2}-n+c$
information qubits. We perform the infinite-depth operations detailed in the
paragraph with\ (\ref{eq:ieid-desired-matrix}) for $c-s$ of the ebits. We then
perform the finite-depth operations detailed in the proofs of this lemma and
Lemma~\ref{lemma:general-ops} in reverse order. The resulting stabilizer has
equivalent error-correcting properties to the quantum check matrix in
(\ref{eq:orig-gens}).

The receiver decodes by first performing all of the finite-depth operations in
the encoding circuit in reverse order. The receiver then decodes the
infinite-depth operations by the procedure listed in the paragraph with
(\ref{eq:decoding-iefd}) so that the original $k_{1}+k_{2}-n+c$ information
qubits per frame are available for processing at the receiving end.%
\end{proof}%

\subsection{Special Case of Entanglement-Assisted Codes with Infinite-Depth
Encoding Circuits and Finite-Depth Decoding Circuits}

We now detail a special case of the above codes in this final section. These
codes are interesting because the information qubits teleport coherently to
other physical qubits when encoding and decoding is complete.

\begin{lemma}
\label{lemma:iefd}Suppose that the Smith form of $F\left(  D\right)  $ in
(\ref{eq:last-eq-fefd}) is%
\[
F\left(  D\right)  =A_{F}\left(  D\right)
\begin{bmatrix}
\Gamma_{F}\left(  D\right)  & 0
\end{bmatrix}
B_{F}\left(  D\right)  ,
\]
where $A_{F}\left(  D\right)  $ is $\left(  n-k_{1}\right)  \times\left(
n-k_{1}\right)  $, $\Gamma_{F}\left(  D\right)  $ is an $\left(
n-k_{1}\right)  \times\left(  n-k_{1}\right)  $\ diagonal matrix whose entries
are powers of $D$, and $B_{F}\left(  D\right)  $ is $k_{2}\times k_{2}$. Then
the resulting entanglement-assisted code admits an encoding circuit with both
infinite-depth and finite-depth operations and admits a decoding circuit with
finite-depth operations only. The information qubits also teleport coherently
to other physical qubits for this special case of codes.
\end{lemma}

%

\begin{proof}%
We perform all the operations in Lemma~\ref{lemma:general-ops} to obtain the
quantum check matrix in (\ref{eq:last-eq-fefd}). Then perform the row
operations in $A_{F}^{-1}\left(  D\right)  $ and the column operations in
$B_{F}^{-1}\left(  D\right)  $. The quantum check matrix becomes%
\[
\left[  \left.
\begin{array}
[c]{ccc}%
E^{\prime}\left(  D\right)  & \Gamma_{F}\left(  D\right)  & 0\\
0 & 0 & 0
\end{array}
\right\vert
\begin{array}
[c]{ccc}%
0 & 0 & 0\\
I & 0 & 0
\end{array}
\right]  ,
\]
where $E^{\prime}\left(  D\right)  =A_{F}^{-1}\left(  D\right)  E\left(
D\right)  $. The Smith form of $E^{\prime}\left(  D\right)  $ is%
\[
E^{\prime}\left(  D\right)  =A_{E^{\prime}}\left(  D\right)
\begin{bmatrix}
\Gamma_{1}\left(  D\right)  & 0 & 0\\
0 & \Gamma_{2}\left(  D\right)  & 0\\
0 & 0 & 0
\end{bmatrix}
B_{E^{\prime}}\left(  D\right)  ,
\]
where $A_{E^{\prime}}\left(  D\right)  $ is $\left(  n-k_{1}\right)
\times\left(  n-k_{1}\right)  $, $\Gamma_{1}\left(  D\right)  $ is an $s\times
s$ diagonal matrix whose entries are powers of $D$, $\Gamma_{2}\left(
D\right)  $ is a $\left(  c-s\right)  \times\left(  c-s\right)  $ diagonal
matrix whose entries are arbitrary polynomials, and $B_{E^{\prime}}\left(
D\right)  $ is $\left(  n-k_{2}\right)  \times\left(  n-k_{2}\right)  $.

Now perform the row operations in $A_{E^{\prime}}^{-1}\left(  D\right)  $ and
the column operations in $B_{E^{\prime}}^{-1}\left(  D\right)  $. It is
possible to counteract the effect of the row operations on $\Gamma_{F}\left(
D\right)  $ by performing column operations, and it is possible to counteract
the effect of the column operations on the identity matrix in the
\textquotedblleft X\textquotedblright\ matrix by performing row operations.
The quantum check matrix becomes%
\[
\left[  \left.
\begin{array}
[c]{ccccccc}%
\Gamma_{1} & 0 & 0 & \Gamma_{1}^{^{\prime}} & 0 & 0 & 0\\
0 & \Gamma_{2} & 0 & 0 & \Gamma_{2}^{^{\prime}} & 0 & 0\\
0 & 0 & 0 & 0 & 0 & \Gamma_{3}^{^{\prime}} & 0\\
0 & 0 & 0 & 0 & 0 & 0 & 0
\end{array}
\right\vert
\begin{array}
[c]{ccc}%
0 & 0 & 0\\
0 & 0 & 0\\
0 & 0 & 0\\
I & 0 & 0
\end{array}
\right]  ,
\]
where $\Gamma_{1}^{^{\prime}}$, $\Gamma_{2}^{^{\prime}}$, and $\Gamma
_{3}^{^{\prime}}$ represent the respective $s\times s$, $\left(  c-s\right)
\times\left(  c-s\right)  $, and $\left(  n-k_{1}-c\right)  \times\left(
n-k_{1}-c\right)  $ diagonal matrices resulting from counteracting the effect
of row operations $A_{E^{\prime}}^{-1}\left(  D\right)  $ on $\Gamma
_{F}\left(  D\right)  $. (We suppress the $D$ argument in all of the matrices
in the above equation.) We use Hadamard and CNOT\ gates to clear the entries
in $\Gamma_{1}^{^{\prime}}\left(  D\right)  $. The quantum check matrix
becomes%
\[
\left[  \left.
\begin{array}
[c]{ccccccc}%
\Gamma_{1} & 0 & 0 & 0 & 0 & 0 & 0\\
0 & \Gamma_{2} & 0 & 0 & \Gamma_{2}^{^{\prime}} & 0 & 0\\
0 & 0 & 0 & 0 & 0 & \Gamma_{3}^{^{\prime}} & 0\\
0 & 0 & 0 & 0 & 0 & 0 & 0
\end{array}
\right\vert
\begin{array}
[c]{ccc}%
0 & 0 & 0\\
0 & 0 & 0\\
0 & 0 & 0\\
I & 0 & 0
\end{array}
\right]  .
\]
The first $s$ rows with $\Gamma_{1}$ and rows $n-k_{1}-c+1,\ldots,n-k_{1}-c+s$
with the identity matrix on the \textquotedblleft X\textquotedblright\ side
stabilize a set of $s$ ebits. The $n-k_{1}-c$ rows with $\Gamma_{3}^{\prime}$
and the $n-k_{2}-c$ rows with identity in the \textquotedblleft
X\textquotedblright\ matrix stabilize a set of $2\left(  n-c\right)
-k_{1}-k_{2}$\ ancilla qubits (up to Hadamard gates). The $s$ and
$k_{2}-n+k_{1}$ columns with zeros in both the \textquotedblleft
Z\textquotedblright\ and \textquotedblleft X\textquotedblright\ matrices
correspond to information qubits. The decomposition of these rows is now complete.

We need to finish processing the $c-s$ rows with $\Gamma_{2}\left(  D\right)
$ and $\Gamma_{2}^{^{\prime}}\left(  D\right)  $ as entries and the $c-s$ rows
of the identity in the \textquotedblleft X\textquotedblright\ matrix. We
construct a submatrix of the above quantum check matrix:%
\begin{equation}
\left[  \left.
\begin{array}
[c]{cc}%
\Gamma_{2}\left(  D\right)  & \Gamma_{2}^{^{\prime}}\left(  D\right) \\
0 & 0
\end{array}
\right\vert
\begin{array}
[c]{cc}%
0 & 0\\
I & 0
\end{array}
\right]  . \label{eq:submatrix}%
\end{equation}

We describe a procedure to encode the above entries with $c-s$ ebits and $c-s$
information qubits using infinite-depth operations. Consider the following
stabilizer matrix%
\begin{equation}
\left[  \left.
\begin{array}
[c]{ccc}%
I & I & 0\\
0 & 0 & 0
\end{array}
\right\vert
\begin{array}
[c]{ccc}%
0 & 0 & 0\\
I & I & 0
\end{array}
\right]  , \label{eq:initial-stab-submatrix}%
\end{equation}
where all identity and null matrices are $\left(  c-s\right)  \times\left(
c-s\right)  $. The above matrix stabilizes a set of $c-s$ ebits and $c-s$
information qubits. Bob's half of the ebits are the $c-s$ columns on the left
in both the \textquotedblleft Z\textquotedblright\ and \textquotedblleft
X\textquotedblright\ matrices and Alice's half are the next $c-s$ columns. We
also track the logical operators for the information qubits to verify that the
circuit encodes and decodes properly. The information-qubit matrix is as
follows%
\begin{equation}
\left[  \left.
\begin{array}
[c]{ccc}%
0 & 0 & I\\
0 & 0 & 0
\end{array}
\right\vert
\begin{array}
[c]{ccc}%
0 & 0 & 0\\
0 & 0 & I
\end{array}
\right]  , \label{eq:init-info-matrix-sub}%
\end{equation}
where all matrices are again $\left(  c-s\right)  \times\left(  c-s\right)  $.
Alice performs Hadamard gates on her first $c-s$ qubits and then performs CNOT
gates from her first $c-s$ qubits to her last $c-s$ qubits to transform
(\ref{eq:initial-stab-submatrix}) to the following stabilizer%
\[
\left[  \left.
\begin{array}
[c]{ccc}%
I & 0 & 0\\
0 & I & 0
\end{array}
\right\vert
\begin{array}
[c]{ccc}%
0 & I & \Gamma_{2}^{^{\prime}}\left(  D\right) \\
I & 0 & 0
\end{array}
\right]  .
\]
The information-qubit matrix in (\ref{eq:init-info-matrix-sub})\ becomes%
\[
\left[  \left.
\begin{array}
[c]{ccc}%
0 & \Gamma_{2}^{^{\prime}}\left(  D^{-1}\right)  & I\\
0 & 0 & 0
\end{array}
\right\vert
\begin{array}
[c]{ccc}%
0 & 0 & 0\\
0 & 0 & I
\end{array}
\right]  .
\]
Alice then performs infinite-depth operations on her last $c-s$ qubits. These
infinite-depth operations correspond to the elements of $\Gamma_{2}%
^{-1}\left(  D\right)  $. She finally performs Hadamard gates on her $2\left(
c-s\right)  $ qubits. The stabilizer becomes%
\begin{equation}
\left[  \left.
\begin{array}
[c]{ccc}%
I & I & \Gamma_{2}^{-1}\left(  D\right)  \Gamma_{2}^{^{\prime}}\left(
D\right) \\
0 & 0 & 0
\end{array}
\right\vert
\begin{array}
[c]{ccc}%
0 & 0 & 0\\
I & I & 0
\end{array}
\right]  , \label{eq:next-stab-submatrix}%
\end{equation}
and the information-qubit matrix becomes%
\begin{equation}
\left[  \left.
\begin{array}
[c]{ccc}%
0 & 0 & 0\\
0 & 0 & \Gamma_{2}^{-1}\left(  D\right)
\end{array}
\right\vert
\begin{array}
[c]{ccc}%
0 & \Gamma_{2}^{^{\prime}}\left(  D^{-1}\right)  & \Gamma_{2}\left(
D^{-1}\right) \\
0 & 0 & 0
\end{array}
\right]  . \label{eq:final-encoded-info-sub}%
\end{equation}
The stabilizer in (\ref{eq:next-stab-submatrix})\ is equivalent to the
following stabilizer by row operations (premultiplying the first $c-s$ rows by
$\Gamma_{2}\left(  D\right)  $):%
\begin{equation}
\left[  \left.
\begin{array}
[c]{ccc}%
\Gamma_{2}\left(  D\right)  & \Gamma_{2}\left(  D\right)  & \Gamma
_{2}^{^{\prime}}\left(  D\right) \\
0 & 0 & 0
\end{array}
\right\vert
\begin{array}
[c]{ccc}%
0 & 0 & 0\\
I & I & 0
\end{array}
\right]  . \label{eq:augmented-submatrix}%
\end{equation}
The measurements that Bob performs have finite weight because the row
operations are multiplications of the rows by the arbitrary polynomials in
$\Gamma_{2}\left(  D\right)  $. Alice thus encodes a code equivalent to the
desired quantum check matrix in\ (\ref{eq:submatrix})\ using $c-s$ ebits and
$c-s$ information qubits.

We now discuss decoding the stabilizer in (\ref{eq:next-stab-submatrix}) and
information qubits. Bob performs CNOTs from the first $c-s$ qubits to the next
$c-s$ qubits. The stabilizer becomes%
\begin{equation}
\left[  \left.
\begin{array}
[c]{ccc}%
0 & I & \Gamma_{2}^{-1}\left(  D\right)  \Gamma_{2}^{^{\prime}}\left(
D\right) \\
0 & 0 & 0
\end{array}
\right\vert
\begin{array}
[c]{ccc}%
0 & 0 & 0\\
I & 0 & 0
\end{array}
\right]  , \label{eq:first-decode-stab-submatrix}%
\end{equation}
and the information-qubit matrix does not change. Bob uses Hadamard and
finite-depth CNOT\ gates to multiply the last $c-s$ columns in the
\textquotedblleft Z\textquotedblright\ matrix by $\Gamma_{2}^{^{\prime}%
}\left(  D^{-1}\right)  \Gamma_{2}\left(  D\right)  $ and add the result to
the middle $c-s$\ columns. It is possible to use finite-depth operations
because the entries of $\Gamma_{2}^{^{\prime}}\left(  D\right)  $ are all
powers of $D$ so that $\Gamma_{2}^{^{\prime}}\left(  D^{-1}\right)
=\Gamma_{2}^{^{\prime}-1}\left(  D\right)  $. The stabilizer in
(\ref{eq:first-decode-stab-submatrix})\ becomes%
\[
\left[  \left.
\begin{array}
[c]{ccc}%
0 & 0 & \Gamma_{2}^{-1}\left(  D\right)  \Gamma_{2}^{^{\prime}}\left(
D\right) \\
0 & 0 & 0
\end{array}
\right\vert
\begin{array}
[c]{ccc}%
0 & 0 & 0\\
I & 0 & 0
\end{array}
\right]  ,
\]
and the information-qubit matrix in (\ref{eq:final-encoded-info-sub})\ becomes%
\[
\left[  \left.
\begin{array}
[c]{ccc}%
0 & 0 & 0\\
0 & \Gamma_{2}^{^{\prime}}\left(  D^{-1}\right)  & \Gamma_{2}^{-1}\left(
D\right)
\end{array}
\right\vert
\begin{array}
[c]{ccc}%
0 & \Gamma_{2}^{^{\prime}}\left(  D^{-1}\right)  & 0\\
0 & 0 & 0
\end{array}
\right]  .
\]
We premultiply the first $c-s$ rows of the stabilizer by $\Gamma_{2}%
^{^{\prime}}\left(  D^{-1}\right)  $ and add the result to the second $c-s$
rows of the information-qubit matrix. These row operations from the stabilizer
to the information-qubit matrix result in the information-qubit matrix having
pure logical operators for the middle $c-s$ qubits. Perform Hadamard gates on
the second set of $c-s$ qubits. The resulting information-qubit matrix is as
follows%
\begin{equation}
\left[  \left.
\begin{array}
[c]{ccc}%
0 & \Gamma_{2}^{^{\prime}}\left(  D^{-1}\right)  & 0\\
0 & 0 & 0
\end{array}
\right\vert
\begin{array}
[c]{ccc}%
0 & 0 & 0\\
0 & \Gamma_{2}^{^{\prime}}\left(  D^{-1}\right)  & 0
\end{array}
\right]  , \label{eq:second-class-info-decoded}%
\end{equation}
so that the information qubits are available at the end of decoding.
Processing may delay or advance them with respect to their initial locations
because the matrix $\Gamma_{2}^{^{\prime}}\left(  D^{-1}\right)  $ is diagonal
with powers of $D$. We can determine that the information qubits teleport
coherently\ from the last set of $c-s$ qubits to the second set of $c-s$
qubits in every frame by comparing (\ref{eq:second-class-info-decoded}) to
(\ref{eq:init-info-matrix-sub}).

The overall procedure for encoding is to begin with a set of $c$ ebits,
$2\left(  n-c\right)  -k_{1}-k_{2}$ ancilla qubits, and $k_{1}+k_{2}-n+c$
information qubits. We perform the infinite-depth operations detailed
in\ (\ref{eq:submatrix}-\ref{eq:augmented-submatrix}) for $c-s$ of the ebits.
We then perform the finite-depth operations detailed in the proofs of this
lemma and Lemma~\ref{lemma:general-ops} in reverse order. The resulting
stabilizer has equivalent error-correcting properties to the quantum check
matrix in (\ref{eq:orig-gens}).

The receiver decodes by first performing all of the finite-depth operations in
reverse order. The receiver then decodes the infinite-depth operations by the
procedure listed in (\ref{eq:first-decode-stab-submatrix}%
-\ref{eq:second-class-info-decoded}) so that the original $k_{1}+k_{2}-n+c$
information qubits per frame are available for processing at the receiving
end.%
\end{proof}%

\begin{example}
\label{ex:brun-example}Consider a classical convolutional code with the
following check matrix:%
\[
H\left(  D\right)  =\left[
\begin{array}
[c]{cc}%
1 & 1+D
\end{array}
\right]  .
\]
We can use the above check matrix in an entanglement-assisted quantum
convolutional code to correct for both bit flips and phase flips. We form the
following quantum check matrix:%
\begin{equation}
\left[  \left.
\begin{array}
[c]{cc}%
1 & 1+D\\
0 & 0
\end{array}
\right\vert
\begin{array}
[c]{cc}%
0 & 0\\
1 & 1+D
\end{array}
\right]  . \label{eq:brun-desired-stab}%
\end{equation}
We first perform some manipulations to put the above quantum check matrix into
a standard form. Perform a CNOT\ from qubit one to qubit two in the same frame
and in the next frame. The above matrix becomes%
\[
\left[  \left.
\begin{array}
[c]{cc}%
D^{-1}+1+D & 1+D\\
0 & 0
\end{array}
\right\vert
\begin{array}
[c]{cc}%
0 & 0\\
1 & 0
\end{array}
\right]  .
\]
Perform a Hadamard gate on qubits one and two. The matrix becomes%
\[
\left[  \left.
\begin{array}
[c]{cc}%
0 & 0\\
1 & 0
\end{array}
\right\vert
\begin{array}
[c]{cc}%
D^{-1}+1+D & 1+D\\
0 & 0
\end{array}
\right]  .
\]
Perform a CNOT\ from qubit one to qubit two. The matrix becomes%
\[
\left[  \left.
\begin{array}
[c]{cc}%
0 & 0\\
1 & 0
\end{array}
\right\vert
\begin{array}
[c]{cc}%
D^{-1}+1+D & D^{-1}\\
0 & 0
\end{array}
\right]  .
\]
Perform a row operation that delays the first row by $D$. Perform a Hadamard
on both qubits. The stabilizer becomes%
\[
\left[  \left.
\begin{array}
[c]{cc}%
1+D+D^{2} & 1\\
0 & 0
\end{array}
\right\vert
\begin{array}
[c]{cc}%
0 & 0\\
1 & 0
\end{array}
\right]  .
\]
The above matrix is now in standard form. The matrix $F\left(  D\right)  =1$
as in (\ref{eq:last-eq-fefd}) so that its only invariant factor is equal to
one. The code falls into the second class of entanglement-assisted quantum
convolutional codes. We begin encoding with one ebit and one information qubit
per frame. The stabilizer matrix for the unencoded stream is as follows:%
\[
\left[  \left.
\begin{array}
[c]{ccc}%
1 & 1 & 0\\
0 & 0 & 0
\end{array}
\right\vert
\begin{array}
[c]{ccc}%
0 & 0 & 0\\
1 & 1 & 0
\end{array}
\right]  ,
\]
and the information-qubit matrix is as follows:%
\[
\left[  \left.
\begin{array}
[c]{ccc}%
0 & 0 & 0\\
0 & 0 & 1
\end{array}
\right\vert
\begin{array}
[c]{ccc}%
0 & 0 & 1\\
0 & 0 & 0
\end{array}
\right]  .
\]
Perform a Hadamard on qubit two and a CNOT\ from qubit two to qubit three so
that the above stabilizer becomes%
\[
\left[  \left.
\begin{array}
[c]{ccc}%
1 & 0 & 0\\
0 & 1 & 0
\end{array}
\right\vert
\begin{array}
[c]{ccc}%
0 & 1 & 1\\
1 & 0 & 0
\end{array}
\right]  ,
\]
and the information-qubit matrix becomes%
\[
\left[  \left.
\begin{array}
[c]{ccc}%
0 & 0 & 0\\
0 & 1 & 1
\end{array}
\right\vert
\begin{array}
[c]{ccc}%
0 & 0 & 1\\
0 & 0 & 0
\end{array}
\right]  .
\]%
\begin{figure*}
[ptb]
\begin{center}
\includegraphics[
natheight=10.380300in,
natwidth=15.453300in,
height=4.2393in,
width=6.3002in
]
{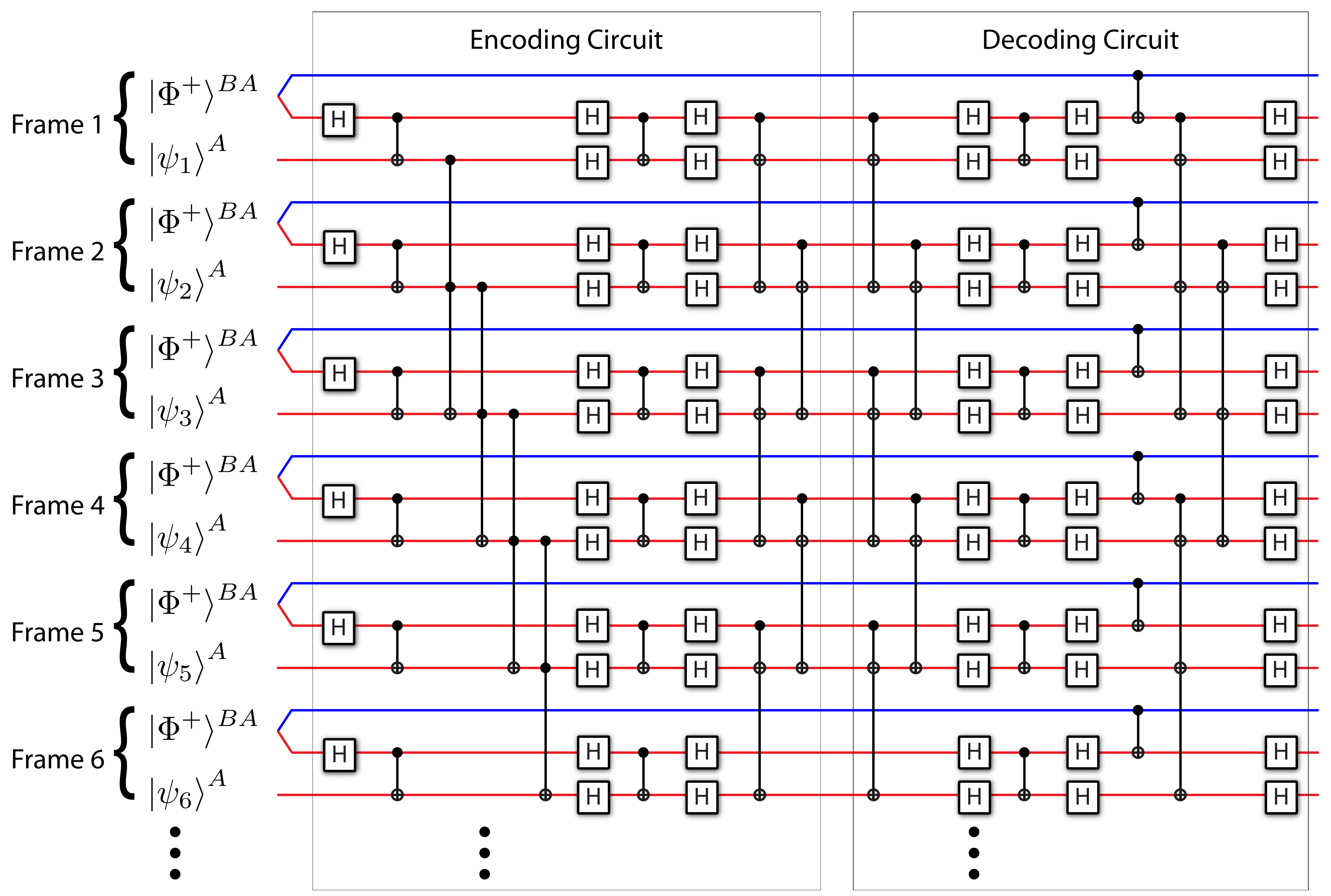}
\caption{(Color online) The encoding and decoding circuits for the
entanglement-assisted quantum convolutional code in
Example~\ref{ex:brun-example}. The third series of gates in the above encoding
circuit is an infinite-depth operation. The other operations in the encoding
circuit are finite-depth. The decoding circuit has finite-depth operations
only.}
\label{fig:example-brun}
\end{center}
\end{figure*}%
Perform an infinite-depth operation corresponding to the rational polynomial
$1/\left(  1+D+D^{2}\right)  $ on qubit three. Follow with a Hadamard gate on
qubits two and three. The stabilizer matrix becomes
\begin{equation}
\left[  \left.
\begin{array}
[c]{ccc}%
1 & 1 & 1/\left(  1+D+D^{2}\right) \\
0 & 0 & 0
\end{array}
\right\vert
\begin{array}
[c]{ccc}%
0 & 0 & 0\\
1 & 1 & 0
\end{array}
\right]  , \label{eq:brun-stab}%
\end{equation}
and the information-qubit matrix becomes%
\begin{equation}
\left[  \left.
\begin{array}
[c]{ccc}%
0 & 0 & 1/\left(  1+D+D^{2}\right) \\
0 & 0 & 0
\end{array}
\right\vert
\begin{array}
[c]{ccc}%
0 & 0 & 0\\
0 & 1 & 1+D^{-1}+D^{-2}%
\end{array}
\right]  . \label{eq:brun-info}%
\end{equation}
Perform the finite-depth operations above in reverse order so that the
stabilizer becomes%
\[
\left[  \left.
\begin{array}
[c]{ccc}%
D^{-1} & \frac{1}{1+D+D^{2}} & \frac{1+D}{1+D+D^{2}}\\
0 & 0 & 0
\end{array}
\right\vert
\begin{array}
[c]{ccc}%
0 & 0 & 0\\
1 & 1 & 1+D
\end{array}
\right]  ,
\]
and the information-qubit matrix becomes%
\[
\left[  \left.
\begin{array}
[c]{ccc}%
0 & \frac{D^{-1}+D^{-2}}{1+D+D^{2}} & \frac{1}{1+D+D^{2}}\\
0 & 0 & 0
\end{array}
\right\vert
\begin{array}
[c]{ccc}%
0 & 0 & 0\\
0 & D^{-1}+D^{-2} & D^{-1}%
\end{array}
\right]  .
\]
The above stabilizer is equivalent to the desired quantum check matrix in
(\ref{eq:brun-desired-stab}) by a row operation that multiplies its first row
by $1+D+D^{2}$.\newline\newline The receiver decodes by performing the
finite-depth encoding operations in reverse order and gets the stabilizer in
(\ref{eq:brun-stab})\ and the information-qubit matrix in (\ref{eq:brun-info}%
). The receiver performs a CNOT\ from qubit one to qubit two and follows with
a CNOT from qubit two to qubit three in the same frame, in an advanced frame,
and in a twice-advanced frame. Finally perform a Hadamard gate on qubits two
and three. The stabilizer becomes%
\[
\left[  \left.
\begin{array}
[c]{ccc}%
0 & 0 & 0\\
0 & 0 & 0
\end{array}
\right\vert
\begin{array}
[c]{ccc}%
0 & 0 & 1/\left(  1+D+D^{2}\right) \\
0 & 0 & 0
\end{array}
\right]  ,
\]
and the information-qubit matrix becomes%
\[
\left[  \left.
\begin{array}
[c]{ccc}%
0 & 0 & 0\\
0 & 1 & 0
\end{array}
\right\vert
\begin{array}
[c]{ccc}%
0 & 1 & 1/\left(  1+D+D^{2}\right) \\
0 & 0 & 0
\end{array}
\right]  .
\]
The receiver decodes the information qubits successfully because a row
operation from the first row of the stabilizer to the first row of the
information-qubit matrix gives the proper logical operators for the
information qubits. Figure~\ref{fig:example-brun}\ details the above encoding
and decoding operations for this entanglement-assisted quantum convolutional code.
\end{example}

\subsection{Discussion}

\label{sec:2nd-class-assumpts}This second class of codes assumes that
noiseless encoding is available. We require this assumption because the
encoding circuit employs infinite-depth encoding operations.

If an error does occur during the encoding process, it can propagate
infinitely through the encoded qubit stream. The result of a single encoding
error can distort both the encoded quantum information, the syndromes that
result from measurements, and the final recovery operations based on the syndromes.

We may be able to relax the noiseless encoding assumption if nearly noiseless
encoding is available. The probability of an error would have to be negligible
in order to ensure that the probability for a catastrophic failure is
negligible. One way to lower the probability of an encoding error is to encode
first with a quantum block code and then further encode with our quantum
convolutional coding method. Many classical coding systems exploit this
technique, the most popular of which is a Reed-Solomon encoder followed by a
convolutional encoder.

\section{Conclusion and Current Work}

This work develops the theory of entanglement-assisted quantum convolutional
coding. We show several methods for importing two arbitrary classical binary
convolutional codes for use in an entanglement-assisted quantum convolutional
code. Our methods outline different ways for encoding and decoding our
entanglement-assisted quantum convolutional codes.

We introduce the notion of an infinite-depth operation for encoding circuits.
We use these infinite-depth operations in both encoding and decoding. These
operations are acceptable if we assume that noiseless processing is available
both at the sender's end and on the receiver's half of shared ebits.

Our first class of codes employs only finite-depth operations in their
encoding and decoding procedures. These codes are the most useful in practice
because they do not have the risk of catastrophic error propagation. An error
that occurs during encoding, measurement, recovery, or decoding propagates
only to a finite number of neighboring qubits.

Our second class of codes uses infinite-depth operations during encoding. This
assumption is reasonable only if noiseless encoding is available. The method
of concatenated coding is one way to approach nearly noiseless encoding in practice.

We suggest several lines of inquiry from here. Our codes are not only useful
for quantum communication, but should also be useful for private classical
communication because of the well-known connection between a quantum channel
and private classical channel \cite{ieee2005dev}. It may make sense from a
practical standpoint to begin investigating the performance of our codes for
encoding secret classical messages. The commercial success of quantum key
distribution for the generation of a private shared secret key motivates this
investigation. It is also interesting to determine which entanglement-assisted
codes can correct for errors on the receiver's side. Codes that possess this
property will be more useful in practice.

The authors thank Hari Krovi and Markus Grassl for useful discussions. They
thank Shesha Raghunathan and Markus Grassl for useful comments on the
manuscript. MMW\ acknowledges support from NSF Grant CCF-0545845,\ and
TAB\ acknowledges support from NSF Grant CCF-0448658.

\bibliographystyle{IEEEtran}
\bibliography{eaqcc-ieee}

\begin{thebibliography}{10}
\providecommand{\url}[1]{#1}
\csname url@samestyle\endcsname
\providecommand{\newblock}{\relax}
\providecommand{\bibinfo}[2]{#2}
\providecommand{\BIBentrySTDinterwordspacing}{\spaceskip=0pt\relax}
\providecommand{\BIBentryALTinterwordstretchfactor}{4}
\providecommand{\BIBentryALTinterwordspacing}{\spaceskip=\fontdimen2\font plus
\BIBentryALTinterwordstretchfactor\fontdimen3\font minus
  \fontdimen4\font\relax}
\providecommand{\BIBforeignlanguage}[2]{{%
\expandafter\ifx\csname l@#1\endcsname\relax
\typeout{** WARNING: IEEEtran.bst: No hyphenation pattern has been}%
\typeout{** loaded for the language `#1'. Using the pattern for}%
\typeout{** the default language instead.}%
\else
\language=\csname l@#1\endcsname
\fi
#2}}
\providecommand{\BIBdecl}{\relax}
\BIBdecl

\bibitem{PhysRevA.52.R2493}
P.~W. Shor, ``Scheme for reducing decoherence in quantum computer memory,''
  \emph{Phys. Rev. A}, vol.~52, no.~4, pp. R2493--R2496, Oct 1995.

\bibitem{PhysRevLett.77.793}
A.~M. Steane, ``Error correcting codes in quantum theory,'' \emph{Phys. Rev.
  Lett.}, vol.~77, no.~5, pp. 793--797, Jul 1996.

\bibitem{PhysRevA.54.1098}
A.~R. Calderbank and P.~W. Shor, ``Good quantum error-correcting codes exist,''
  \emph{Phys. Rev. A}, vol.~54, no.~2, pp. 1098--1105, Aug 1996.

\bibitem{thesis97gottesman}
D.~Gottesman, ``Stabilizer codes and quantum error correction,'' Ph.D.
  dissertation, California Institue of Technology, 1997.

\bibitem{PhysRevLett.78.405}
A.~R. Calderbank, E.~M. Rains, P.~W. Shor, and N.~J.~A. Sloane, ``Quantum error
  correction and orthogonal geometry,'' \emph{Phys. Rev. Lett.}, vol.~78,
  no.~3, pp. 405--408, Jan 1997.

\bibitem{ieee1998calderbank}
A.~Calderbank, E.~Rains, P.~Shor, and N.~Sloane, ``Quantum error correction via
  codes over gf(4),'' \emph{IEEE Trans. Inf. Theory}, vol.~44, pp. 1369--1387,
  1998.

\bibitem{PhysRevLett.79.3306}
P.~Zanardi and M.~Rasetti, ``Noiseless quantum codes,'' \emph{Phys. Rev.
  Lett.}, vol.~79, no.~17, pp. 3306--3309, Oct 1997.

\bibitem{mpl1997zanardi}
------, ``Error avoiding quantum codes,'' \emph{Modern Physics Letters B},
  vol.~11, pp. 1085--1093, 1997.

\bibitem{PhysRevLett.81.2594}
D.~A. Lidar, I.~L. Chuang, and K.~B. Whaley, ``Decoherence-free subspaces for
  quantum computation,'' \emph{Phys. Rev. Lett.}, vol.~81, no.~12, pp.
  2594--2597, Sep 1998.

\bibitem{kribs:180501}
D.~Kribs, R.~Laflamme, and D.~Poulin, ``Unified and generalized approach to
  quantum error correction,'' \emph{Physical Review Letters}, vol.~94, no.~18,
  p. 180501, 2005.

\bibitem{qic2006kribs}
D.~W. Kribs, R.~Laflamme, D.~Poulin, and M.~Lesosky, ``Operator quantum error
  correction,'' \emph{Quantum Information \& Computation}, vol.~6, pp.
  383--399, 2006.

\bibitem{poulin:230504}
D.~Poulin, ``Stabilizer formalism for operator quantum error correction,''
  \emph{Physical Review Letters}, vol.~95, no.~23, p. 230504, 2005.

\bibitem{isit2007brun}
T.~Brun, I.~Devetak, and M.-H. Hsieh, ``General entanglement-assisted quantum
  error-correcting codes,'' in \emph{IEEE International Symposium on
  Information Theory}, June 2007.

\bibitem{arxiv2007brun}
M.-H. Hsieh, I.~Devetak, and T.~A. Brun, ``General entanglement-assisted
  quantum error-correcting codes,'' \emph{arXiv:0708.2142}, 2007.

\bibitem{book2007mermin}
N.~D. Mermin, \emph{Quantum Computer Science}.\hskip 1em plus 0.5em minus
  0.4em\relax Cambridge University Press, 2007.

\bibitem{nat1982}
W.~K. Wootters and W.~H. Zurek, ``A single quantum cannot be cloned,''
  \emph{Nature}, vol. 299, pp. 802--803, 1982.

\bibitem{book1983code}
F.~J. MacWilliams and N.~J.~A. Sloane, \emph{The Theory of Error-Correcting
  Codes}.\hskip 1em plus 0.5em minus 0.4em\relax North Holland, 1983.

\bibitem{PhysRevA.66.052313}
G.~Bowen, ``Entanglement required in achieving entanglement-assisted channel
  capacities,'' \emph{Phys. Rev. A}, vol.~66, no.~5, p. 052313, Nov 2002.

\bibitem{PhysRevLett.70.1895}
C.~H. Bennett, G.~Brassard, C.~Cr\'epeau, R.~Jozsa, A.~Peres, and W.~K.
  Wootters, ``Teleporting an unknown quantum state via dual classical and
  einstein-podolsky-rosen channels,'' \emph{Phys. Rev. Lett.}, vol.~70, no.~13,
  pp. 1895--1899, Mar 1993.

\bibitem{PhysRevLett.69.2881}
C.~H. Bennett and S.~J. Wiesner, ``Communication via one- and two-particle
  operators on einstein-podolsky-rosen states,'' \emph{Phys. Rev. Lett.},
  vol.~69, no.~20, pp. 2881--2884, Nov 1992.

\bibitem{arx2006brun}
T.~A. Brun, I.~Devetak, and M.-H. Hsieh, ``Catalytic quantum error
  correction,'' \emph{arXiv:quant-ph/0608027}, August 2006.

\bibitem{science2006brun}
------, ``Correcting quantum errors with entanglement,'' \emph{Science}, vol.
  314, no. 5798, pp. pp. 436 -- 439, October 2006.

\bibitem{PhysRevA.55.1613}
S.~Lloyd, ``Capacity of the noisy quantum channel,'' \emph{Phys. Rev. A},
  vol.~55, no.~3, pp. 1613--1622, Mar 1997.

\bibitem{capacity2002shor}
P.~W. Shor, ``The quantum channel capacity and coherent information,'' in
  \emph{Lecture Notes, MSRI Workshop on Quantum Computation}, 2002.

\bibitem{ieee2005dev}
I.~Devetak, ``The private classical capacity and quantum capacity of a quantum
  channel,'' \emph{IEEE Trans. on Inf. Theory}, vol.~51, pp. 44--55, January
  2005.

\bibitem{PhysRevLett.83.3081}
C.~H. Bennett, P.~W. Shor, J.~A. Smolin, and A.~V. Thapliyal,
  ``Entanglement-assisted classical capacity of noisy quantum channels,''
  \emph{Phys. Rev. Lett.}, vol.~83, no.~15, pp. 3081--3084, Oct 1999.

\bibitem{ieee2002bennett}
------, ``Entanglement-assisted capacity of a quantum channel and the reverse
  shannon theorem,'' \emph{IEEE Trans. on Inf. Theory}, vol.~48, pp.
  2637--2655, 2002.

\bibitem{book1991cover}
T.~M. Cover and J.~A. Thomas, \emph{Elements of Information Theory}.\hskip 1em
  plus 0.5em minus 0.4em\relax Wiley-Interscience, 1991.

\bibitem{PhysRevLett.91.177902}
H.~Ollivier and J.-P. Tillich, ``Description of a quantum convolutional code,''
  \emph{Phys. Rev. Lett.}, vol.~91, no.~17, p. 177902, Oct 2003.

\bibitem{arxiv2004olliv}
------, ``Quantum convolutional codes: fundamentals,''
  \emph{arXiv:quant-ph/0401134}, 2004.

\bibitem{isit2006grassl}
M.~Grassl and M.~R\"{o}tteler, ``Noncatastrophic encoders and encoder inverses
  for quantum convolutional codes,'' in \emph{IEEE International Symposium on
  Information Theory (quant-ph/0602129)}, 2006.

\bibitem{ieee2006grassl}
------, ``Quantum convolutional codes: Encoders and structural properties,'' in
  \emph{Forty-Fourth Annual Allerton Conference}, 2006.

\bibitem{ieee2007grassl}
------, ``Constructions of quantum convolutional codes,'' in \emph{IEEE Int.
  Symp. on Inf. Theory}, 2007.

\bibitem{isit2005forney}
G.~D. Forney and S.~Guha, ``Simple rate-1/3 convolutional and tail-biting
  quantum error-correcting codes,'' in \emph{IEEE International Symposium on
  Information Theory (arXiv:quant-ph/0501099)}, 2005.

\bibitem{ieee2007forney}
G.~D. Forney, M.~Grassl, and S.~Guha, ``Convolutional and tail-biting quantum
  error-correcting codes,'' \emph{IEEE Trans. Inf. Theory}, vol.~53, pp.
  865--880, 2007.

\bibitem{book2000mikeandike}
M.~A. Nielsen and I.~L. Chuang, \emph{Quantum Computation and Quantum
  Information}.\hskip 1em plus 0.5em minus 0.4em\relax Cambridge University
  Press, 2000.

\bibitem{PhysRevA.54.1862}
D.~Gottesman, ``Class of quantum error-correcting codes saturating the quantum
  hamming bound,'' \emph{Phys. Rev. A}, vol.~54, no.~3, pp. 1862--1868, Sep
  1996.

\bibitem{arx2007wilde}
M.~M. Wilde, H.~Krovi, and T.~A. Brun, ``Convolutional entanglement
  distillation,'' \emph{arXiv:0708.3699}, 2007.

\bibitem{luo:010303}
Z.~Luo and I.~Devetak, ``Efficiently implementable codes for quantum key
  expansion,'' \emph{Phys. Rev. A}, vol.~75, no.~1, p. 010303, 2007.

\bibitem{pra2007wildeEA}
M.~M. Wilde, H.~Krovi, and T.~A. Brun, ``Entanglement-assisted quantum error
  correction with linear optics,'' \emph{Phys. Rev. A}, vol.~76, p. 052308,
  2007.

\bibitem{prep2007wildeEAOQEC}
M.~M. Wilde and T.~A. Brun, ``Protecting quantum information with entanglement
  and noisy optical modes,'' \emph{In preparation}, 2007.

\bibitem{arx2005dev}
I.~Devetak, A.~W. Harrow, and A.~Winter, ``A resource framework for quantum
  shannon theory,'' \emph{arXiv:quant-ph/0512015}, 2005.

\bibitem{itit1967viterbi}
A.~J. Viterbi, ``Error bounds for convolutional codes and an asymptotically
  optimum decoding algorithm,'' \emph{IEEE Trans. Inf. Theory}, vol.~13, pp.
  260--269, 1967.

\bibitem{book1999conv}
R.~Johannesson and K.~S. Zigangirov, \emph{Fundamentals of Convolutional
  Coding}.\hskip 1em plus 0.5em minus 0.4em\relax Wiley-IEEE Press, 1999.

\bibitem{prl2004harrow}
A.~Harrow, ``Coherent communication of classical messages,'' \emph{Phys. Rev.
  Lett.}, vol.~92, p. 097902, March 2004.

\bibitem{wilde:060303}
M.~M. Wilde, H.~Krovi, and T.~A. Brun, ``Coherent communication with continuous
  quantum variables,'' \emph{Phys. Rev. A}, vol.~75, no.~6, p. 060303(R), 2007.

\end{thebibliography}

\end{document}